\documentclass[11pt]{article}
\usepackage[top=1in,bottom=1in,left=1in,right=1in]{geometry}
\usepackage{amssymb,amsmath,enumerate,float}
\usepackage[colorlinks,linkcolor=blue,anchorcolor=blue,citecolor=blue]{hyperref}
\usepackage[dvips]{graphicx}
\usepackage[dvips]{graphics}
\usepackage[dvips]{epsfig}
\usepackage{theorem}

\newtheorem{theorem}{Theorem}[section]
\newtheorem{corollary}[theorem]{Corollary}
\newtheorem{proposition}[theorem]{Proposition}
\theorembodyfont{\upshape}
\newtheorem{definition}[theorem]{Definition}

\theorembodyfont{\upshape}

\newtheorem{remark}[theorem]{Remark}

\newcommand{\proof}{{\noindent\it Proof.}\quad}
\newcommand{\finproof}{{\hfill$\Box$\\}}
\newcommand{\red}{\color{red}}
\newcommand{\blue}{\color{blue}}
\newcommand{\ddr}{\mathrm{d}}

\def\bproposition{\begin{proposition}}\def\eproposition{\end{proposition}}
\def\beqlb{\begin{eqnarray}}\def\eeqlb{\end{eqnarray}}
\def\beqnn{\begin{eqnarray*}}\def\eeqnn{\end{eqnarray*}}
\def\ar{\!\!\!&}

 \def\ar{\!\!\!&}

 \def\mbb{\mathbb}

 \def\proof{\noindent{\it Proof.~~}}\def\qed{\hfill$\Box$\medskip}

\begin{document}

\title{Alpha-CIR Model with Branching Processes in Sovereign Interest Rate Modelling}
\author{Ying Jiao\thanks{Universit\'e Claude Bernard-Lyon 1, Institut de Science Financier et d'Assurances. Email: ying.jiao@univ-lyon1.fr.}, \,\,Chunhua Ma\thanks{Nankai University, School of Mathematical Sciences. Eamil: mach@nankai.edu.cn.}, \,\,Simone Scotti\thanks{Universit\'e Paris Diderot-Paris 7, Laboratoire de Probabilit\'es et Mod\`eles Al\'eatoires. Email: scotti@math.univ-paris-diderot.fr.}}
\maketitle

\abstract{We introduce a class of interest rate models, called the $\alpha$-CIR model, which gives a natural extension of the standard CIR model by adopting the $\alpha$-stable L\'evy process and preserving the branching property.  This model allows to describe in a unified and parsimonious way several recent observations on the  sovereign bond market such as the persistency of low interest rate together with the presence of large jumps at local extent. We emphasize on a general integral representation of the model by using random fields, with which we establish the link to the CBI processes and the affine models. Finally we analyze the jump behaviors and in particular the large jumps, and we provide numerical illustrations.}

\section{Introduction}

On the current European sovereign bond market, there exists a number of well-established and seemingly puzzling facts. On the one hand, the interest rate has reached a historically low level in the Euro countries. However, on the other hand, the sovereign  bond can have very large variations  when uncertainty about unpredictable political or economical events increases, such as in the Greek case. The aim of this paper is to present a new model of interest rate, called the $\alpha$-CIR model, where we give a natural extension of the well-known Cox-Ingersoll-Ross (CIR, see \cite{CIR85}) model by using the $\alpha$-stable branching processes, in order to describe these recent observations on the bond market. In particular, the set of questions investigated includes  the clustering behavior of the variance of sovereign interest rates, and also the  persistency of low interest rates  together with the significant fluctuations  at a local extent.

In the literature, large fluctuations in financial data 
motivate naturally the introduction of jumps in the interest rate dynamics, such as in Eberlein and Raible \cite{EberleinRaible1999}, Filipovi\'c, Tappe and Teichmann \cite{FilipovicTappeTeichmann2010}.
Nevertheless, the jump presence conflicts in general with the trend of low rates, at least as long as the jump intensity is assumed as the paradigm. One way to reconcile large fluctuations with low rate persistency is to use a regime change framework 
but this may increase the dimension of the stochastic processes in order to preserve the Markov property. Recently, the Hawkes processes or the self-exciting point processes (see Hawkes \cite{Hawkes}), have been used to overcome this difficulty since they exhibit properties which give a suitable interpretation of such modelling. A Hawkes process can be seen as a population process whose reproduction rate is proportional to the population itself, 
that is, the so-called self-exciting property. Moreover, the external arrival of migrant can be modeled by  a second point process. 
A large and growing literature is devoted to the financial application of Hawkes processes, in particular, to the interest rate and credit intensity modelling, such as in A\"it-Sahalia, Cacho-Diaz and Laeven \cite{AJ15}, Errais, Giesecke and Goldberg \cite{EGG2010}, Dassios and Zhao \cite{DZ2011} and Rambaldi, Pennesi and Lillo \cite{RPL14}. 
In the above mentioned papers, as appear naturally in Hawkes framework, the driving process is at least two-dimensional since both the
dynamics of jump process and its intensity are taken into account. 

In this paper, 
we introduce a short interest rate model by using the $\alpha$-stable L\'evy processes,  which provides a relatively simple jump diffusion model to respond to these modelling challenges in an endogenous way. We exploit an integral representation of the $\alpha$-CIR model to highlight the branching property. First of all,  branching processes arise as the limit of Hawkes processes and exhibit, by their inherent nature, the clustering or the self-exciting property implying that the jump frequency increases  with the value of the process itself. By consequence, branching processes, thanks to the infinite divisibility of their law with respect to the starting point, prove to be a prolific subject in probability having interesting applications in finance, see for instance Duffie, Filipovi\'c and Schachermayer \cite{DFS2003}. In the modelling of interest rate, branching processes have already been considered  by the pioneering paper of 
Filipovi\'c \cite{F01} where the relationship between the exponential affine structure of bond prices and the branching property has been highlighted. Moreover, our model is a natural generalization of the CIR model  which appears to be the simplest and most popular continuous-time branching process. Although CIR model has closed-form solutions for bond prices which turns out to be a main feature in view of model calibration, it does not include jumps.  In addition, empirical studies underline that the behavior the bond prices cannot be fully explained by CIR model which systematically overestimates short interest rates (see for instance Brown and Dybvig \cite{BD86} and Gibbons and Ramaswamy \cite{GR93}). In our framework, CIR process is the departing model which is the only example of branching process with continuous path and the inclusion of the $\alpha$-stable branching processes allows to better describe the low  interest rate behavior. 

The main contribution of the present paper is to combine the properties of Hawkes and CIR processes in order to   
define the $\alpha$-CIR model, which provides a larger class of jump-diffusion models having the branching property and preserving the explicit expression for bond prices. 
The $\alpha$-CIR model consists of, besides the  Brownian motion, a spectrally positive $\alpha$-stable L\'evy-process. The parameter $\alpha\in(1,2]$ characterizes the tail fatness and the jump behavior. When $\alpha$ equals $2$, the $\alpha$-stable process reduces to a Brownian motion and we recover the classical CIR model. In the general case when $\alpha\in(1,2)$, there may appear infinitely many  jumps in a finite time interval, which represent the fluctuations related to sovereign risks.  
In order to keep the branching property, the square root in the volatility term have to be replaced by the $\alpha$-root of the process. 
 Despite its simplicity and the reduced number of extra parameters compared to the usual CIR, 
the model we develop show several advantages. 
{First}, the $\alpha$-CIR model exhibits positive jumps and, in particular, by combining heavy-tailed jump size distribution with infinite activity, can describe in a 
unified way both the large fluctuations observed in financial market and the usual small oscillations.
{Second}, in a branching process framework can also be shown that a hierarchical structure for interest rate naturally arises, since it can split the interest rate
into different components, which can eventually be interpreted as spreads, each one following the same dynamics, in the similar way in which a global ideal population can be split into subgroups evolving according the same dynamics.  {Third}, by the link established between the $\alpha$-CIR model and the 
 continuous state branching process with immigration (CBI process), we deduce, using the joint Laplace transform of the CBI process, the bond prices in an explicit way. In particular, we show the interesting result that the bond price increases with the tail fatness (that is, decreases with the parameter $\alpha$), which better responds to the persistency of low interest rate behavior of the current sovereign rate. {Fourth}, we can make a thorough analysis of the jump behavior, in particular, for the large jumps which signify in the interest rate dynamics a sudden increasing sovereign risk and imply, for example in the Greek case, a potentially high probability of default.  We are  particularly interested in the first time that such a large jump occurs and explore the impact of the tail index $\alpha$.

We begin our analysis by presenting an equivalence between two different formulations of the dynamics for the $\alpha$-CIR model. {From the theoretical point of view, this property has been thoroughly exploited by Li \cite{Li11} and Li and Ma \cite{LiMa08}.  In the spirit of the above papers,} we shall prove that the usual version of CIR dynamics and its $\alpha$-CIR extension admit an alternative representation which is of integral form by using random fields but  the dimension of the L\'evy basis  has to be increased, for instance 
the Brownian motion is replaced by a two-dimensional white noise.
In the financial literature on interest rates, this approach has already been performed, see for example  Kennedy \cite{Ken}, Albeverio, Lytvynov and Mahnig \cite{ALM}
where random field modelling is introduced to describe the interest rate term structure.
The integral representation allows to better identify the process features like the branching property, and is more convenient for proving related properties. 
Moreover, it needs to be remarked that the integral representation enlightens the  relation between the Ornstein-Uhlenbeck and CIR dynamics,
and then between the L\'evy-Ornstein-Uhlenbeck (LOU) and $\alpha$-CIR models.
As a matter of fact, in an analogous way that the most natural extension of Ornstein-Uhlenbeck dynamics including branching property is the CIR, 
the $\alpha$-CIR results from the combination of the LOU model with $\alpha$-stable driver and the branching property. 

The main, and perhaps most interesting, forecast of the present model is that the bond prices decrease with the parameter $\alpha$, which in turn 
is inversely related to the tail fatness. 
The explanation of this apparently paradoxical result is based on the features of the $\alpha$-CIR model highlighted previously. 
The use of fat-tail distributed positive jumps will imply a large negative compensator, then between two jumps the mean reversion term is magnified whenever $\alpha$ decreases. This phenomenon is the consequence of compensation and the final result is to make both tails heavier.
In general,  the standard behavior of bond prices increases  with respect to the fatness of tails, 
such as the case in ordinary LOU dynamics (see e.g. Barndorff-Nielsen and Shephard \cite{BS01}). 
However, for a given value of $\alpha$,  the branching property adds a new phenomenon in the present case with $\alpha$-CIR model:  the frequency of big jumps decreases when interest rates are low thanks to the self-exciting structure and this allows some ``freezing'' effect of low rates for relatively longer time period. In addition, the strong mean-reverting term resulting in the case of fat-tailed jump distribution will also increase the likelihood of occurrence of the persistency of low rates.  
   
 The paper is organized as follows.  Section  \ref{sec:Model} deals with the mathematical presentation of the $\alpha$-CIR model. 
Section \ref{section-CBI} is devoted to the characterization of our model as a CBI process and the properties derived from this link. 
 In Section \ref{sec:interest-rate}, we apply our model to term structure modeling and exhibit in particular the closed-form bond prices up to a numerical integration. 
Section \ref{section-jumps} deals with the analysis of jumps.  We enrich our results with some numerical illustrations in Section \ref{sec:numerics}.
 Finally, Section \ref{sec:conclusion} concludes the paper.

\section{Model framework}
\label{sec:Model}

This section introduces  the $\alpha$-CIR interest rate model and its basic properties. We start by defining 
two representations of our model and establish an explicit link  between the two classes, so that the properties of each class are directly transferred to the other one. 
Let us fix a probability space $(\Omega, \mathcal F, \mathbb P)$ equipped with a filtration $\mathbb F=(\mathcal F_t)_{t\geq 0}$ 
satisfying the usual conditions.  

\begin{definition}[Root representation]\label{def-SDE-root}
We consider the following diffusion for the short interest rate $r=(r_t,t\geq 0)$ with

\begin{equation}\label{lambda-root}
r_t = r_0 + \int_0^t a \left( b  - r_s  \right) ds + \sigma\int_0^t  \sqrt{r_s} dB_s
+\sigma_Z \int_0^t   {r_{s-}^{1/\alpha}}  dZ_s \end{equation}
where $B=(B_t,t\geq 0)$  is a Browinan motion and $Z=(Z_t,t\geq 0)$ is a spectrally positive $\alpha$-stable compensate L\'evy process with parameter $\alpha \in (1,2]$, which is independent of $B$ and  whose Laplace transform is given, for $q\geq 0$, by
 \beqnn
 \mathbb{E}\left[e^{-qZ_t}\right]=\exp\left\{-\frac{tq^\alpha}{\cos(\pi\alpha/2)}\right\}.
 \eeqnn
In other words,  $Z_t$ follows the $\alpha$-stable distribution with scale parameter $t^{1/\alpha}$,  skewness parameter $1$ and zero drift. , i.e., $Z_t\sim S_\alpha(t^{1/\alpha}, 1,0)$.
 \end{definition}

 We call  processes defined by \eqref{lambda-root} the $\alpha$-CIR processes of parameters $(a,b,\sigma,\sigma_Z,\alpha)$ and denote  by $\alpha$-$\mathrm{CIR}(a,b,\sigma,\sigma_Z,\alpha)$ the set of all such processes. The existence of a unique strong solution of the equation \eqref{lambda-root} follows from Fu and Li \cite[Theorem 5.3]{FL10}. 

It is easy to see that the CIR model belongs to the class in Definition \ref{def-SDE-root} by taking $\sigma_Z =0$.   Another case where we recover a CIR process is when $\alpha=2$. In this case, the process $Z$ becomes a standard  Brownian motion  scaled by the coefficient $\sqrt{2}$ which is independent of $B$. Hence an $\alpha$-CIR process satisfying  \eqref{lambda-root} is actually a CIR process of the form
\[r_t=r_0+\int_0^ta(b-r_s)\,ds+\sqrt{\sigma^2+ 2\sigma_Z^2}\int_0^t\sqrt{r_s}d\widetilde{B}_s\]
where 
$\widetilde{B}=
(\sigma B+\sigma_Z Z)/\sqrt{\sigma^2+2\sigma_Z^2}$ is a standard Brownian motion. In other words, an $\alpha$-CIR process with parameter $\alpha=2$ is a CIR process.


 The departure of the process ${Z}$ from Brownian motion is controlled by the tail index $\alpha$. When $\alpha<2$, ${Z}$ is a pure jump process with heavy tails. For any fixed $t$, the distribution of $Z_t$ is a stable distribution
 and the tail of the distribution decays like a power function with index $-\alpha$. This means that a stable random variable exhibits more
 variability than a Gaussian one and it is more likely to take values far away from the median. Compared to a standard Poisson or compound Poisson process, this pure jump process has an infinite number of (small) jumps over any time interval, allowing it to capture the extreme activity. 
In the meantime, the $\alpha$-stable processes share similar properties with the Brownian motion such as  self-similarity or stability property, which means that the distribution of the $\alpha$-stable process over any horizon has the same shape upon scaling. From the statistical point of view, the process given by (\ref{lambda-root}) is characterized by two more parameter with respect to CIR model, i.e. $\alpha$ and $\sigma_Z$.

We then introduce a more general form of the $\alpha$-CIR model by using random fields. 
\begin{definition}[Integral representation]\label{def-SDE-integral}

We also consider the following equation in the integral form \begin{equation}\label{lambda-integral}
r_t = r_0 + \int_0^t a \left( b  - r_s  \right) ds + \sigma \int_0^t \int_0^{r_s} W(ds,du)
+ \sigma_Z \int_0^t \int_0^{r_{s-} } \int_{\mathbb{R}^+} \zeta \widetilde{N}
(ds,du, d\zeta),\quad t\geq 0
\end{equation}
where $W(ds,du)$ is a white noise on $\mathbb{R}_+^2$  with intensity $dsdu$, $\widetilde N(ds,du,d\zeta)$ is an independent compensated Poisson 
random measure on $\mathbb{R}_+^3$ with intensity $dsdu\mu(d\zeta)$ with $\mu(d\zeta)$ being a L\'evy measure on $\mathbb{R}_+$ and satisfying  
$\int_0^\infty (\zeta\wedge\zeta^2)\mu(d\zeta)<\infty $.  
\end{definition}

We call the process given by \eqref{lambda-integral} the $\alpha$-CIR type process with parameters $(a,b,\sigma,\sigma_Z,\mu)$. It follows from  of Dawson and Li \cite[Theorem~3.1]{DawsonLi} or Li and Ma \cite[Theorem~2.1]{LM13} that the equation \eqref{lambda-integral}  has a unique strong solution. 

We establish a first link to the $\alpha$-CIR model. Let the L\'evy measure $\mu$  be as  \beqlb\label{Levymeasure} 
\mu_\alpha(d\zeta)=-{1_{\{\zeta>0\}} d\zeta \over \cos(\pi\alpha/2)\Gamma(-\alpha)\zeta^{1+\alpha}}, \quad 1<\alpha<2,
 \eeqlb
then the solution of \eqref{lambda-integral} has the same probability law as that of the equation \eqref{lambda-root}. In an extended probability space, for any couple $(B, {Z})$ there exists a couple  $(W,\widetilde{N})$ such that the solution of the two equations 
(\ref{lambda-integral}) and (\ref{lambda-root}) are equal almost surely; see Propositions \ref{1-2} and \ref{2-1} below.

\begin{remark}\label{remark2}
We  explain the connection of the above integral representation to Hawkes processes. We begin by considering  an integral representation of the CIR model. Let  $W(ds,du)$ be a white noise on $\mathbb{R}_+^2$  with intensity $dsdu$. The CIR process $r$ (when $\sigma_Z=0$) is given in  the form 
$r_t = r_0 + \int_0^t a \left( b  - r_s  \right) ds + \sigma \int_0^t \int_0^{r_s} W(ds,du),$
or equivalently as  
\beqlb\label{new form}
r_t=r^{\ast}_t+\sigma\int_0^t\int_0^{r_s}e^{-a(t-s)}W(ds,du)
\eeqlb
where $r^{\ast}_t$ is a deterministic function given by $r^{\ast}_t=r_0e^{-at}+ab\int_0^te^{-a(t-s)}ds$.  The expression \eqref{new form} shows the self-exciting feature.

We then consider a simple Hawkes process with exponential kernel, which is defined as a point process $J$ with intensity $r$, where $r$ reads 
 \beqnn
 r_t = r^{\ast}_t+
 \int_0^t e^{-a(t-s)}dJ_s
\eeqnn
and $r^{\ast}$ is the background rate, i.e., the deterministic part of the process $J$.
When a jump arrives, the intensity $r$ increases, which also increases the probability of a next jump, that is the self-exciting property of Hawkes processes. 
In order to facilitate the comparison with our integral representation, we give a different characterization of the intensity $r$. 
Let $N$ be a Poisson process on $\mathbb{R}^2$ with characteristic measure $dsdu$, so  $J_t$ can be written as the form of $\int_0^t\int_0^{r_{s-}}N(ds,du)$ and $r_t$ as
 \beqlb\label{hawkes2}
 r_t = r^{\ast}_t+
 \int_0^t\int_0^{r_{s-}} e^{-a(t-s)}N(ds,du).
\eeqlb
In this form, the self-exciting feature can be observed as follows: the frequency of jumps grows with the process itself due
to the presence of the integral with respect to the variable $u$. Moreover, when $r^{\ast}$ takes certain particular form, $r$  is a branching process, also known as an affine process in finance (see \cite{DFS2003}). 
In this context, the self-exciting features is equivalent to the branching property. 

Let us now come back to the integral representation \eqref{lambda-integral} of $\alpha$-CIR model. We let $\sigma=0$ and $\mu(d\zeta)=\delta_1(dz)$,  then the (non-compensated) Poisson measure $N(ds,du,d\zeta)$ reduces to a random measure on $\mathbb{R}_+^2$ with intensity $dsdu$, denoted by $N(ds,du)$. Hence $r$  can be rewritten as
 \beqnn
 r_t = r_0 + abt- \int_0^t (a+\sigma_Z)r_s  ds +
\sigma_Z \int_0^t \int_0^{r_{s-} } {N}
(ds,du).
\eeqnn
We note that $r$ is the intensity of the Hawkes process $\int_0^t\int_0^{r_{s-}}N(ds,du)$ by using the equivalent form   
\beqlb\label{hawkes}
r_t=r_0e^{-(a+\sigma_Z)t}+\frac{ab}{a+\sigma_Z}\left(1-e^{(a+\sigma_Z)t}\right)+\int_0^t \int_0^{r_{s-}}e^{-(a+\sigma_Z)(t-s)}N(ds,du).
\eeqlb
As a consequence, $\alpha$-CIR type processes, and in particular the $\alpha$-CIR processes, can be seen as marked Hawkes processes influenced by a Brownian noise.

Furthermore consider a sequence of processes 
$\big\{r_t^{(n)}, t\geq 0\big\}$ defined by (\ref{hawkes}) with parameters $(a/n, nb, \sigma_Z)$. Note that as $n\rightarrow\infty$, we have
\beqnn
r^{(n)}_{nt}/n\overset{\mathcal{L}}{\longrightarrow}Y_t, \quad\mbox{in\ } D(\mathbb{R}_+),
\eeqnn
where $Y$ follows a {CIR model} given by 
$Y_t=\int_0^ta(b-Y_s)ds+\sigma_Z\int_0^t\int_0^{Y_s}W(ds,du) $
and $D(\mathbb{R}_+)$ denotes the c\`{a}dl\`{a}g processes space equipped with the 
Skorokhod topology. Therefore, a sequence of rescaled Hawkes processes converges  weakly to the CIR process, see Jaisson and Rosenbaum \cite{JR15} for more details, notably on the convergence of the nearly unstable Hawkes process with general kernel, after suitably rescaling, to a CIR process. 
\end{remark}

We now develop the equivalence between the root representation in Definition \ref{def-SDE-root} and the integral one in Definition \ref{def-SDE-integral} with the L\'evy measure $\mu_\alpha$. The following two propositions show both implications. { The main idea follows \cite[Theorem 9.32]{Li11}.}

\begin{proposition}\label{1-2}
Let $r$ be a solution to \eqref{lambda-integral} with $\mu=\mu_\alpha$ given by (\ref{Levymeasure}). On an extended probability space of $(\Omega,\mathcal F,\mathbb P)$, there exists a L\'evy process $(B,Z)$ valued in $\mathbb R^2$  where $B$ is a Brownian motion and $Z$ is a spectrally positive $\alpha$-stable compensated L\'evy process, such that $r$ is a solution to \eqref{lambda-root}.
\end{proposition} 
\proof 
We extend the probability space to include a standard Brownian motion $\widehat{B}$ and a spectrally positive $\alpha$-stable compensated L\'evy process $\widehat{Z}$ with L\'evy measure $\mu_\alpha$ as in \eqref{Levymeasure}, such that $\widehat{B}$, $\widehat{Z}$, $W$ and $\widetilde N$ are mutually independent. We then construct the processes $B$ and $Z$ as 
\[B_t=\int_0^t r_s^{-1/2}1_{\{r_s>0\}}\int_0^{r_s}W(ds,du)+\int_0^t1_{\{r_s=0\}}d\widehat{B}_s,\quad t\geq 0\]
and
\[Z_t=\int_0^t r_{s-}^{-1/\alpha}1_{\{r_{s-}>0\}}\int_0^{\infty}\int_0^{r_{s-}}\zeta \widetilde N(ds,du,d\zeta)+\int_0^t 1_{\{r_{s-}=0\}}d\widehat{Z}_s.\] We let $\mathbb F=(\mathcal F_t,\,t\geq 0)$ be the filtration generated by these processes.
Fix $\theta$, ${\theta'}\in\mathbb{R}$. Applying It\^{o}'s formula to the two-dimensional martingale $(B,Z)$ we have for $T>t\geq 0$,
\beqnn
\ar \ar e^{i(\theta B_T+\theta' Z_T)}-e^{i(\theta B_t+\theta' Z_t)}\\
\ar= \ar {(M_T-M_t)}-\frac{\theta^2}{2}\int_t^T e^{i(\theta B_s+\theta' Z_s)}ds\\
\ar\ar+
\frac{1}{\cos(\pi\alpha/2)\Gamma(-\alpha)}\int_t^T e^{i(\theta B_s+
\theta' Z_{s-})}1_{\{r_{s-}=0\}}
\int_0^\infty (e^{i\theta'\zeta}-1-i\theta'\zeta){d\zeta\over\zeta^{1+\alpha}}ds\\
\ar\ar
+
\frac{1}{\cos(\pi\alpha/2)\Gamma(-\alpha)}
\int_t^Te^{i(\theta B_s+\theta' Z_{s-})}r_{s-}1_{\{r_{s-}>0\}}\int_0^\infty (e^{i\theta' r^{-1/\alpha}_{s-}}-1-i\theta' r^{-1/\alpha}_{s-})
{d\zeta\over\zeta^{1+\alpha}}ds\\
\ar=\ar (M_T-M_t)+\left[\frac{1}{\cos(\pi\alpha/2)\Gamma(-\alpha)}\int_0^\infty (e^{i\theta'\zeta}-1-i\theta'\zeta){d\zeta\over\zeta^{1+\alpha}}-\frac{\theta^2}{2}\right]
\int_t^T e^{i(\theta B_s+\theta' Z_s)}ds
\eeqnn
where $M$ is a martingale.
Then multiplying both sides of the above equality by $e^{-i(\theta B_t+\theta' Z_t)}$ and taking conditional expectation, we have that $h_t(T):=\mathbb{E}[  e^{i(\theta (B_T-B_t)+\theta' (Z_{T}-Z_t))}|\mathcal{F}_t]$ satisfies the integral equation
\beqnn
h_t(T)=1+\left[\frac{(\theta')^\alpha}{\cos(\pi\alpha/2)}e^{-i\pi\alpha/2}-\frac{\theta^2}{2}\right]
\int_t^T h_t(s)ds, \quad\mbox{a.s.}
\eeqnn
Solving the above  equation we obtain 
\beqnn
\mathbb{E}\left[\left. e^{i(\theta (B_t-B_l)+\theta'(Z_{t}-Z_l))}\right|\mathcal{F}_l\right]=
\exp\left\{(t-l)\Big(\frac{(\theta')^\alpha}{\cos(\pi\alpha/2)}e^{-i\pi\alpha/2}-\frac{\theta^2}{2}\Big)\right\},
\eeqnn
which implies that $B$ is a standard Brownian motion and $Z$ is a spectrally positive $\alpha$-stable compensated L\'evy process independent of $B$. Moreover, by construction $r$ is a solution to  \eqref{lambda-root}.\qed



\begin{proposition}\label{2-1}
Let $r$ be a solution to  \eqref{lambda-root}. On an expanded probability space of $(\Omega,\mathcal F,\mathbb P)$, there exist a  white noise $W$ on $\mathbb R_+^2$ and a compensated Poisson random measure $\widetilde N$ on $\mathbb R_+^3$ with the L\'evy measure $\mu_\alpha$ given in \eqref{Levymeasure}, which are independent, and such that $r$ verifies \eqref{lambda-integral}.
\end{proposition}
\proof The L\'{e}vy-It\^{o} representation of $Z$ implies that 
$Z_t= \int_0^t\int_0^\infty \zeta \widetilde{N}(ds,d\zeta)$,
where $\widetilde{N}(ds,dz)$ is a compensated Poisson random measure on
$\mathbb{R}_+^2$ with intensity $ds\mu(d\zeta)$ given by (\ref{Levymeasure}). Furthermore, on an extended probability space there exist a white noise
$W{_1}(ds,du)$ on $\mathbb{R}_+\times (0,1)$ with intensity $dsdu$ and a Poisson random measure $N_1(ds,du,d\zeta)$ on $\mathbb{R}_+\times (0,1)\times\mathbb{R}_+$ with intensity $dsdu\mu(d\zeta)$ independent of $W_1$ such that (c.f. El Karoui and M\'{e}l\'{e}ard \cite[Corrollary {\rm
III-5}] {KM90} and
Ikeda and Watanabe \cite[Theorem 6.7]{IW89})
\beqlb\label{2-1B}
B_t=\int_0^t\int_0^1W_1(ds,du)\ \mbox{and}\ \int_0^t\int_{A}N(ds,d\zeta)=\int_0^t\int_0^1\int_{A}N_1(ds,du,d\zeta),
\eeqlb
for $A\in\mathcal{B}(\mathbb{R}_+)$. In a similar way, on an extended probability space let $W_2(du,ds)$ be a white noise on $\mathbb{R}_+^2$  with intensity $dsdu$ and $N_2(ds,du,d\zeta)$ be a Poisson 
random measure on $\mathbb{R}_+^3$ with intensity $dsdu{\mu}(d\zeta)$ independent of $W_2$. Then we define
for any $A, C\in\mathcal{B}(\mathbb{R}_+)$,
\beqlb\label{2-1W}
W([0,t]\times A):=\int_0^t\int_0^1\sqrt{r_s} 1_{A}(r_su)W_1(ds,du)+\int_0^t\int_{r_s}^\infty 1_A(u)W_2(ds,du),
\eeqlb
\beqlb\label{2-1N}
N([0,t]\times A\times C)\ar:=\ar\int_0^t\int_0^1\int_0^\infty 1_A(r_{s-}u)1_C({r_{s-}^{1/\alpha}}\zeta)N_1(ds,du,d\zeta)\nonumber\\
\ar\ar+\int_0^t\int_{r_{s-}}^\infty\int_0^\infty1_A(u)1_C(\zeta)N_2(ds,du,d\zeta).
\eeqlb
Similarly as in Proposition \ref{1-2}, $(W,N)$ has the same distribution as $(W_2,N_2)$. So  the proposition is proved.\qed

The branching property is one key property of the $\alpha$-CIR model. The following result shows that the $\alpha$-CIR process $r$ has the 
branching property in the pathwise sense, {see   \cite[Theorem 3.2]{DawsonLi}}. The proof is based on the integral representation \eqref{lambda-integral} where the white noise $W$ and the compensated Poisson random measure $\widetilde N$ are translation invariant with respect to the variable $u$.  
\begin{proposition}\label{prop:decomposition}
 Let $r$ be an $\alpha$-{CIR}$(a,b,\sigma,\sigma_Z,\alpha)$ process. Let  $r_0^{(i)}\in\mathbb R_+$ and $b^{(i)}\in\mathbb R$,  $i\in\{1,2\}$, such that $r_0=r_0^{(1)}+r_0^{(2)}$ and $b=b^{(1)}+b^{(2)}$. Then there exist independent processes $r^{(i)}$ in the families $\alpha$-{CIR}$(a,b^{(i)},\sigma,\sigma_Z,\alpha)$ with initial values $r_0^{(i)}$ such that $r=r^{(1)}+r^{(2)}$. 
\end{proposition}
\begin{proof}
Let $r$ be a solution to \eqref{lambda-integral} with L\'evy measure $\mu_\alpha$. Define $r^{(1)}$ to be the solution to the following  equation  \begin{equation}\label{Equ:lambda1}
r_t^{(1)} = r_0^{(1)} 
+ \int_0^t a \left( b  - r_s^{(1)}  \right) ds +
 \sigma \int_0^t \int_0^{r_s^{(1)}} W(ds, du)
+ \sigma_Z \int_0^t \int_0^{r_{s-}^{(1)} } \int_{\mathbb{R}^+} \zeta \widetilde{N}
(ds,du, d\zeta).
\end{equation} where $(W,N)$ are the same as in \eqref{lambda-integral}.
Note that  $r^{(1)}$ is an $\alpha$-CIR process with parameters $(a,b^{(1)},\sigma,\sigma_Z,\alpha)$.  By \cite[Theorem~3.2]{DawsonLi}, we have for all
$t\geq 0$, $\mathbb P(r_t\geq r_t^{(1)})=1$.  Let $r^{(2)}=r-r^{(1)}$. Then
\[r_t^{(2)}=r^{(2)}_0  + \int_0^t a \left( b^{(2)}  - r^{(2)}_s  \right) ds + \sigma \int_0^t 
\int_{r^{(1)}_s}^{r^{(1)}_s+ r^{(2)}_s} W(ds, du)  + \sigma_Z \int_0^t \int_{r^{(1)}_{s-}}^{r^{(1)}_{s-}+ r^{(2)}_{s-} } \int_{\mathbb{R}^+} 
\zeta \widetilde{N}(ds,du, d\zeta).\]
By the translation invariance of $W$ and $\widetilde{N}$ with respect to the variable $u$, we obtain that $r^{(2)}$ is independent of $r^{(1)}$  and is an $\alpha$-CIR process with parameters $(a,b^{(2)},\sigma,\sigma_Z,\alpha)$. The proposition is thus proved.\qed
\end{proof}


To study the effect of the branching property, we introduce the locally equivalent L\'evy-Ornstein-Uhlenbeck (LOU) process to make a comparison with the $\alpha$-CIR process.

\begin{definition}[Locally equivalent LOU process]\label{def-LevyOU}
Let $\lambda=(\lambda_t,t\geq 0)$ be the solution of the following equation
\begin{equation}\label{Levy-OU-integral}
\lambda_t = r_0 + \int_0^t a \left( b  - \lambda_s  \right) ds + \sigma \int_0^t \int_0^{r_0} W(ds,du)
+ \sigma_Z \int_0^t \int_0^{r_0} \int_{\mathbb{R}^+} \zeta \widetilde{N}
(ds,du, d\zeta),
\end{equation} 
where the initial value $r_0$, and the processes $W$ and $\widetilde N$ are the same as in Definition \ref{def-SDE-integral}.
\end{definition}

Note  that, in the case where the L\'evy measure is given by $\mu_\alpha$, the process $\lambda$ defined by \eqref{Levy-OU-integral} can be written in the following form as a generalization of the Vasicek model
\begin{equation}\label{lou process}
\lambda_t=r_0+ \int_0^t a \left( b  - \lambda_s  \right) ds +\sigma \sqrt{r_0}B_t+\sigma_Z\sqrt[\alpha]{r_0}Z_t,
\end{equation}
where $B$ and $Z$ are the same as in Definition \ref{def-SDE-root}. 
Comparing  \eqref{lambda-integral} and  \eqref{Levy-OU-integral}, we remark that at the initial time,  the two  processes have the same volatility and jump terms. But when time evolves,  the volatility and jump terms of $\alpha$-CIR process will be adapted to the actual level of the interest rate, while ``frozen'' at the initial values in the locally equivalent LOU process.

To further study the difference between  \eqref{lambda-integral} and \eqref{Levy-OU-integral}, we separate the large and small jumps  and use the non-compensated version of the Poisson random measure $\widetilde N$. Since $\alpha$-stable processes exhibit infinite activity, we fix a
jump threshold $y$ (so the threshold for $r$ is given as $\overline{y}=\sigma_Z y$).
The small jumps with infinite activity can be approximated by a second Brownian motion for instance in the spirit of  Asmussen and Rosinski \cite{AR2001}.  
The locally equivalent LOU process reads

\beqlb\label{Levy-OU-integral-without-compensation}
\lambda_t \ar=\ar r_0 + \int_0^t a \left(b-\frac{\sigma_Z {r_0} \Theta(\alpha, y)}{a} - \lambda_s\right)ds + 
 \sigma\int_0^t \int_0^{r_0} W(ds,du)\nonumber\\
\ar\ar
+\sigma_Z \int_0^t \int_0^{r_0 } \int_0^y\zeta \widetilde{N}
(ds, du, d\zeta)
 +   \sigma_Z \int_0^t \int_0^{r_0 } \int_y^\infty \zeta N
(ds,du,d\zeta)\, ,
\eeqlb
where
\begin{equation}\label{def-Theta}
\Theta(\alpha, y) =-\frac{1}{\cos(\pi\alpha/2)\Gamma(-\alpha)} \int_{y}^{\infty}\frac{d\zeta}{\zeta^{\alpha}} 
= \frac{2}{\pi}\alpha\Gamma(\alpha-1)\sin(\pi\alpha/2)y^{-(\alpha-1)},
\end{equation}
and $N$ is the (non-compensated) Poisson random measure corresponding to $\widetilde N$.
In a similar way, the $\alpha$-CIR process \eqref{lambda-integral} can be written as
\begin{equation}\label{lambda integral-without-compensation}
\begin{array}{rcl}
\displaystyle r_t &=& \displaystyle r_0 + \int_0^t \widetilde{a}(\alpha, y) \left( \widetilde{b}(\alpha, y)  - r_s  \right) ds 
+ \sigma \int_0^t \int_0^{r_s } W(ds, du)  \\
&&\quad  \displaystyle + \sigma_Z \int_0^t \int_0^{r_{s-} } \int_0^y \zeta \widetilde{N}
(ds, du, d\zeta) +  \sigma_Z \int_0^t \int_0^{r_{s-} } \int_y^\infty \zeta N
(ds, du, d\zeta)\, ,
\end{array}
\end{equation}
where 
\begin{equation}\label{ab}
\widetilde{a}(\alpha, y) = a+ \sigma_Z\Theta(\alpha, y) \quad \quad  
\widetilde{b}(\alpha, y) = \frac{ab}{a+ \sigma_Z \Theta(\alpha, y)}
\end{equation}

The previous results allow us to make comparisons. 
First, comparing $\alpha$-CIR and LOU processes, it follows  that the implicit negative drifts from large jump part lead to a linear decay for $\lambda_t$ 
while to a stronger exponential decay for $r_t$. 
Then as $\sigma_Z$ increases, the decreasing drift term plays a more important role in $r_t$ than in $\lambda_t$. 
 Second, comparing CIR and $\alpha$-CIR processes, namely the cases $\sigma_Z=0$ and
$\sigma_Z>0$ in (\ref{lambda integral-without-compensation}), we can study the evolution between two large jumps in the $\alpha$-CIR model.
Between two large jumps, the $\alpha$-CIR exhibits an increasing mean reverting speed $\widetilde{a}$ and a decreasing long mean interest rate $\widetilde{b}$
as long as $\sigma_Z$ increases. As a consequence, the $\alpha$-CIR diffusion is more adequate to model the presence of low interest rates and their persistency upon large jumps, compared to LOU and CIR models.

We are also interested in the jump times of large jumps. For this purpose, we introduce, based on (\ref{lambda integral-without-compensation}), the auxiliary process
 \beqlb\label{rhat1}
\widehat{r}^{(y)}_t =r_0 + \int_0^t \widetilde{a}(\alpha, y) \big( \widetilde{b}(\alpha, y)  - r_s  \big) ds 
+ \sigma \int_0^t \int_0^{r_s} W(ds,du)
+ \sigma_Z \int_0^t \int_0^{r_{s-} } \int_0^y\zeta \widetilde{N}
(ds,du,d\zeta).
\eeqlb
For any jump threshold $y>0$, the process $\widehat{r}^{(y)}$ coincides with $r$ up to the first large jump  $\tau_1^{(y)}:=\inf\{t>0:\Delta r_t>\overline y=\sigma_Z y\}$. More generally, denote by $\{\tau_i^{(y)}\}_{i\in\mathbb{N}}$ the sequence of  jump times of $r$ larger than $\overline{y}$, 
then for any $t\in[\tau_i^{(y)}, \tau_{i+1}^{(y)})$, we have
 \begin{equation}\label{rhat2}
 \begin{array}{rcl}
\displaystyle \widehat{r}^{(y)}_t &=& \displaystyle r_{\tau_i^{(y)}}  + \int_{\tau_i^{(y)}}^t
 \widetilde{a}(\alpha, y) \big( \widetilde{b}(\alpha, y)  - \widehat{r}^{(y)}_s  \big) ds 
+ \sigma \int_{\tau_i^{(y)}}^t \int_0^{\widehat{r}^{(y)}_s} W(ds,du) \\
&& \displaystyle + \sigma_Z \int_{\tau_i^{(y)}}^t \int_0^{\widehat{r}^{(y)}_{s-} } \int_0^y\zeta \widetilde{N}
(ds,du,d\zeta).
\end{array}
\end{equation}
This auxiliary process $\widehat{r}^{(y)}$ represents the history of the interest rate $r$ except the jumps larger than $\overline y$. 
This process will be particularly useful to study $\tau^{(y)}_1$ (see Section \ref{section-jumps}).



\section{Link with the CBI processes}\label{section-CBI}

In this section, we show that $\alpha$-CIR processes are continuous state branching processes with immigration (CBI process) and 
deduce from this fact several properties of the $\alpha$-CIR model. 
The CBI processes have been introduced by Kawazu and Watanabe \cite{KaW71}. We recall the definition as below.

\begin{definition}[CBI process]
A Markov process
$X$ with state space $\mbb{R}_+$ is called a {continuous state branching process with 
immigration,}  characterized by
\emph{branching mechanism} $\Psi(\cdot)$ and \emph{immigration rate} $\Phi(\cdot)$, if its   
characteristic representation is 
given, for $p\geq 0$, by
\begin{equation}\label{laplace}
\mathbb{E}_{x}\left[e^{- p X_{t}}\right]=\exp\left[-xv(t,p)-\int_{0}^{t}\Phi\big(v(s,p)\big)ds\right],
\end{equation} 
where the  function $v:\mathbb R_+\times\mathbb R_+\rightarrow\mathbb R$ satisfies the following differential equation 
\begin{equation}\label{ODE0}\frac{\partial v(t,p)}{\partial t}=-\Psi(v(t,p)),\quad v(0,p)=p\end{equation}
and $\Psi$ and $\Phi$ are functions of the variable $q\geq 0$ given by
\beqnn
\Psi(q)&=&\beta q+\frac{1}{2}\sigma^{2}q^{2}+\int_{0}^{\infty}(e^{-qu}-1+qu)\pi(du), \\
\Phi(q)&=&\gamma q+\int_{0}^{\infty}(1-e^{-qu})\nu(du),
\eeqnn
 with $\sigma, \gamma \geq 0$, $\beta \in \mathbb{R}$ and  $\pi$, $\nu$ being two L\'evy measures such that 
\beqlb\label{levy measure moment}
\int_{0}^{\infty} (u\wedge u^{2})\pi(du)<\infty,\quad \int_{0}^{\infty} (1\wedge u) \nu(du)<\infty.
\eeqlb
\end{definition}

 The CBI process $X$ has as generator the operator $\mathcal{L}$ acting on $C^{2}_{0}(\mathbb{R}_{+})$ as 
\beqlb\label{generator1}
\mathcal{L} f(x)&=&\frac{\sigma^{2}}{2}xf''(x)+(\gamma-\beta x)f'(x)+x\int_{0}^{\infty}(f(x+u)-f(x)-uf'(x))\pi(du)\nonumber\\
&&\;+\int_{0}^{\infty}\left(f(x+u)-f(x)\right)\nu(du).
\eeqlb 

The next proposition shows that the $\alpha$-CIR model belongs to the family of CBI processes by using the integral representation  \eqref{lambda-integral}, {see \cite[Theorem 3.1]{DawsonLi}}. 
We shall give two proofs. The first one in the main text is by verification. The second one, which is constructive, is postponed in Appendix.
\begin{proposition}\label{pro: integral CBI}
The $\alpha$-CIR type process $r$ in Definition \ref{def-SDE-integral} is a CBI process with the branching mechanism $\Psi$ given by
 \begin{equation}\label{equ: Psi general}
 \Psi(q)=aq +\frac{1}{2}\sigma^2q^2+ \int_0^{\infty}(e^{-q\sigma_Z\zeta}-1+q\sigma_Z\zeta){\mu}(d\zeta)
 \end{equation}
and the immigration rate $\Phi(q)=a b q$.

\end{proposition}


\smallskip
\proof 
Let $v(t,p)$ be the unique solution of the differential equation \eqref{ODE0} with $\Psi$ given by \eqref{equ: Psi general}.
 Fix $t>0$ and let $u(s,p)=v(t-s,p)$ for $0\leq s\leq t$.  Denote by   
 $$Y_s^{(p)}:=\exp\Big(-u(s,p)r_s+ab\int_0^s u(l,p)dl\Big)$$ 
 Applying It\^{o}'s formula, we have that 
 \beqnn
Y_t^{(p)}-Y_0^{(p)}-\int_0^tr_sY_s^{(p)}\Big(\Psi(u(s,p))-\frac{\partial{u(s,p)}}{\partial{s}}\Big)ds, \,t\geq 0  \eeqnn
is a martingale. By (\ref{ODE0}), 
$(Y_t^{(p)}-Y_0^{(p)},t\geq 0)$ is a martingale. Thus $
\mathbb E[Y_t^{(p)}]=e^{-u(0,p)r_0}$, which implies that 
 \beqnn
 \mathbb{E}[e^{-pr_t}]=\exp\Big({-v(t,p)r_0-ab\int_0^t v(t-s,p)ds}\Big)=\exp\Big({-v(t,p)r_0-ab\int_0^t v(s,p)ds}\Big). \eeqnn
Moreover, since $r$ is the unique strong solution of equation (\ref{lambda-integral}), it is a Markov process. 
 Thus $r$ is a CBI $(\Psi,\Phi)$ process.\qed

As consequence of the previous proposition, the $\alpha$-CIR model and its truncated process are both CBI processes by considering particular L\'evy measures. 

 \begin{corollary}\label{pro: SCIR CBI}The  $\alpha$-CIR $(a,b,\sigma,\sigma_Z,\alpha)$ process is a CBI process with the branching mechanism $\Psi$ given by
 \begin{equation}\label{equ: Psi SCIR}
 \Psi_{\alpha}(q)=a q+\frac{\sigma^2}{2}q^2-\frac{\sigma_Z^\alpha}{\cos(\pi\alpha/2)}q^\alpha,
 \end{equation}
and the immigration rate $\Phi$ given by
\begin{equation}\label{equ: Phi SCIR}
\Phi(q)=a b q.
\end{equation} 
\end{corollary}

 \begin{corollary}\label{pro: auxiliary CBI}The auxiliary process $ \widehat{r}^{(y)}$ defined by \eqref{rhat1} is a CBI process with the branching mechanism $ \Psi^{(y)}$ given by
 \begin{equation}\label{Psix}
 \Psi^{(y)}_{\alpha}(q):=\Big(a+\sigma_Z^\alpha\int_y^\infty\zeta\mu_\alpha(d\zeta)\Big)q+\frac{1}{2}\sigma^2q^2+\sigma_Z^\alpha\int_0^y(e^{-q\zeta}-1+q\zeta)\mu_\alpha(d\zeta),
 \end{equation}
where $\mu_\alpha$ is given by \eqref{Levymeasure} and the immigration rate $\Phi$ given by
 \begin{equation}\label{Phi auxiliary}
 \Phi(q)=\widetilde{a}(\alpha, y)\,  \widetilde{b}(\alpha, y)\,  q = a b \,q.
 \end{equation} 
\end{corollary}

In the following of this section, we use the CBI characterization to show some properties of the $\alpha$-CIR model.

 
\begin{proposition} \label{continuity} Let $(r^{(\alpha)}_t,t\geq0)$ denote the $\alpha$-CIR process with parameters $(a,b,\sigma,\sigma_Z,\alpha)$. Then as $\alpha\rightarrow2$,  $r^{(\alpha)}$ converges in distribution in $D(\mathbb{R}_+)$ to the CIR process $r^{(2)}$.
\end{proposition}

\proof For any $\alpha\in(1,2]$, the $\alpha$-CIR process $r^{(\alpha)}$ is a CBI  process. Let $P^{(\alpha)}$ be the transition semigroup of $r^{(\alpha)}$ and $A^{(\alpha)}$ be its generator.  Denote  $e_p(x)=e^{-px}$ for $p>0$ and $x\geq0$. Then by (\ref {generator1}), 
$$A^{(\alpha)}e_p(x)=-e_p(x)\left(x\Psi_\alpha(p)+\Phi(p)\right)=-e^{-px}\Big(x\big(ap+\frac{\sigma^2}{2}p^2-\frac{\sigma_Z^{\alpha}}{\cos(\pi\alpha/2)}p^{\alpha}\big)+abp\Big).$$
We have \beqnn
\lim_{\alpha\rightarrow 2}\,\sup_{x\in\mathbb{R}_+}|A^{(\alpha)}e_p(x)-A^{(2)}e_p(x)|=0.
\eeqnn
Denote by $D_1$ the linear hull of 
$\{e_p:p>0\}$. Then $D_1$ is an algebra which strongly separates the points of $\mathbb{R}_+$. 
Let $C_0(\mathbb{R}_+)$ be the space of continuous functions on $\mathbb{R}_+$ vanishing at infinity. 
By the Stone-Weierstrass theorem, $D_1$ is dense in $C_0(\mathbb{R}_+)$. Since $D_1$ is invariant under $P^{(2)}$ (see (\ref{laplace})), it is a core of $A^{(2)}$ (see Ethier and Kurtz \cite[Proposition 3.3]{EK86}). 
By \cite[Corollary 8.7]{EK86}, we have the weak convergence of the processes as $\alpha$ tends to $2$. 
\qed
 
\begin{proposition}\label{stationary}
The  $\alpha$-CIR type process in Definition \ref{def-SDE-integral}  has a limit distribution, whose Laplace transform is given by
 \beqlb\label{limiting distribution}
 \mathbb{E}[e^{-pr_{\infty}}]=\exp\Big(-\int_0^p\frac{\Phi(q)}{\Psi(q)}dq\Big),\quad p\geq 0.
 \eeqlb
Moreover, the process is exponentially ergodic, namely  $\|\mathbb P(r_t\in\,\cdot\,)-\mathbb P(r_\infty\in\,\cdot\,)\|\leqslant C\rho^t$ for some positive constants $C$ and $\rho<1$, where $\|\cdot\|$ denotes the total variation norm.  \end{proposition}
\proof 
Note that the branching mechanism $\Psi$ is bounded from below by $aq+\frac 12\sigma^2q^2$. Hence one has
\[\int_0^1\frac{\Phi(q)}{\Psi(q)}dq\leqslant\int_0^1\frac{abq}{aq+\frac 12\sigma^2q^2}dq<\infty.\]
By \cite[Theorem 3.20]{Li11}, we obtain that the process $r$ in Definition \ref{def-SDE-integral} has a limit distribution, whose  Laplace transform is given by
$\exp\big(-\int_0^\infty\Phi(v(t,p))dt\big)$,
where the function $v$ is defined in \eqref{ODE0}. A change of variables $q=v(t,p)$ in the above formula leads to \eqref{limiting distribution}.  The last assertion follows from \cite[Theorem 2.5]{LM13}.\qed

Finally, we show that the usual condition of inaccessibility of the point $0$ is preserved when we extend CIR model to the $\alpha$-CIR one.
 
\begin{proposition}\label{inaccessible 0}
For the  $\alpha$-CIR $(a,b,\sigma,\sigma_Z,\alpha)$ process with $\alpha\in(1,2)$, the point $0$ is an inaccessible boundary  if and only if $2a b\geq\sigma^2$. In particular, a pure jump $\alpha$-CIR process  with $ab>0$ never reaches $0$.   
 \end{proposition} 
 \proof We apply the result of Duhalde, Foucart and Ma \cite[Theorem 2]{{DFM14}} for CBI processes to obtain that 
 $0$ is an inaccessible boundary point for an $\alpha$-CIR type process  if and only if 
\beqlb\label{integral condition}
  \int_\theta^\infty \frac{dz}{\Psi(z)}\exp\Big(\int_\theta^z\frac{\Phi(x)}{\Psi(x)}dx\Big)=\infty
  \eeqlb
 for some positive constant $\theta$, where $\Psi$ is given by (\ref{equ: Psi general}) and $\Phi(q)=abq$. We now focus on the $\alpha$-CIR process.  Let $\Psi^*(q)=aq+\sigma^2q^2/2$ be the branching mechanism of the classical CIR process viewed as a CBI process. One has $\Psi_\alpha\geq\Psi^*$, where $\Psi_\alpha$ is the branching mechanism of the $\alpha$-CIR process, given in \eqref{equ: Psi SCIR}. Therefore 
\[\int_\theta^\infty\frac{dz}{\Psi_\alpha(z)}\exp\Big(\int_\theta^z\frac{\Phi(x)}{\Psi_\alpha(x)}dx\Big)\leq\int_\theta^\infty\frac{dz}{\Psi^*(z)}\exp\Big(\int_\theta^z\frac{\Phi(x)}{\Psi^*(x)}dx\Big).\]
In particular, if $0$ is an inaccessible boundary for the $\alpha$-CIR$(a,b,\sigma,\sigma_Z,\alpha)$ process, then the inequality $2ab\geq\sigma^2$ holds, thanks to the classical inaccessibility criterion for the CIR processes.

Conversely, if the inequality $2ab\geq\sigma^2$ holds, then one has
\[\frac{\Phi(x)}{\Psi_\alpha(x)}=\frac{1}{x}(1+O(x^{\alpha-2})),\quad x\rightarrow+\infty.\]
So there exists a constant $C>0$ (depending on $\theta$) such that 
\[\int_\theta^z\frac{\Phi(x)}{\Psi_\alpha(x)}dx\geq\log(z/\theta)-C.\]
Hence
\[\int_\theta^\infty\frac{dz}{\Psi_\alpha(z)}\exp\Big(\int_\theta^z\frac{\Phi(x)}{\Psi_\alpha(x)}dx\Big)\geqslant\frac{1}{e^C\theta}\int_\theta^\infty\frac{z}{\Psi_\alpha(z)}dz=+\infty.\]
  \qed 

\begin{remark}
The result of Proposition \ref{inaccessible 0} is not true when $\alpha=2$. In this case the $\alpha$-CIR model reduces to a classical CIR model, but with a modified volatility term. Therefore for the $\alpha$-CIR$(a,b,\sigma,\sigma_Z,2)$ process, the point $0$ is an inaccessible boundary if and only if $2ab\geq\sigma^2+2\sigma_Z^2$. We note that when the $\alpha$-CIR process contains the jump part,  the parameter $\sigma_Z$ does not intervene in the boundary condition by the above proposition. 
\end{remark}

\section{Applications to interest rate modeling}\label{sec:interest-rate}

In this section, we apply the $\alpha$-CIR model to the interest rate modelling and pricing. 
Since the $\alpha$-CIR model is a generalization of  the classical CIR model
by adding jumps but preserving the CBI properties, the bond price has an affine structure, see Filipovi\'c \cite{F01}. We give a closed-form expression of the bond price depending on a function which is the integral of the reciprocal of $1- \Psi_\alpha$. This integral can be easily computed numerically, and so is semi-explicit formula for the Laplace transform of the integrated interest rate. Moreover, we show that the bond price is decreasing with respect to the index parameter $\alpha$.
In the next part, we focus on a path dependent option, more precisely a put option written on the running minimum of the bond yield. We show that the payoff of this option can be rewritten as a put option written on the running minimum of the spot rate itself with different nominal and strike. Despite the non-Markovian behavior and the non linearity of the payoff, its price can be obtained  by inversion of the 
Laplace transform.

\subsection{Zero-coupon bond pricing}  

We begin by making precise the equivalent probability measures.
The following proposition shows that the short interest rate $r$, given by the $\alpha$-CIR model, remains to be in the class of integral type processes under an equivalent change of probability.  

\begin{proposition}\label{pro:changementprob} Let $r$ be  an $\alpha$-CIR$(a,b,\sigma,\sigma_Z,\alpha)$ processes under the probability measure $\mathbb P$ and assume that the  filtration $\mathbb F$ is generated by the random fields $W$ and $\widetilde N$.  Fix 
$\eta\in\mathbb{R}$ and $\theta\in\mathbb{R}_+$, and define
 \beqnn
 U_t:=\eta\int_0^t\int_0^{r_s}W(ds,du)+\int_0^t\int_0^{r_{s-}}\int_0^\infty (e^{-\theta\zeta}-1)\widetilde{N}(ds,du,d\zeta).
 \eeqnn
Then the Dol\'eans-Dade exponential $\mathcal{E} (U)$ 
is a martingale and the probability 
measure $\mathbb Q$ defined by 
 \beqnn
\left. \frac{d\mathbb Q}{d\mathbb P}\right|_{\mathcal{F}_t}=\mathcal{E} (U)_t,
 \eeqnn
 is equivalent to $\mathbb P$. Moreover, under $\mathbb Q$, $r$ is an $\alpha$-CIR type process with the parameters 
$(a',b',\sigma',\sigma_Z',\mu_\alpha')$, where 
$$a'=a-\sigma\eta-\frac{\alpha\sigma_Z}{\cos(\pi\alpha/2)}\theta^{\alpha-1}, \quad b'=ab/a', \quad \sigma'=\sigma,\quad \sigma_Z'=\sigma_Z$$ and 
\beqnn
\mu_\alpha'(d\zeta)=-\frac{e^{-\theta\zeta}}{\cos(\pi\alpha/2)\Gamma(-\alpha)\zeta^{1+\alpha}}d\zeta.
\eeqnn
\end{proposition}
\proof  The couple $(r, U)$ is a time homogeneous affine process (c.f. \cite[Theorem 6.2]{DL06}).
The Dol\'eans-Dade exponential $\mathcal{E} (U)$ is a true martingale 
by checking that the conditions in \cite[Corollary 3.2]{KMK2010} are satisfied, so it defines an equivalent probability measure $\mathbb Q$.
Note that $Y=\mathcal{E} (U)$ is the unique strong solution of $dY_t=Y_{t-}dU_t$. Then for any function  $f\in C^2(\mathbb{R_+})$, the process  
\beqnn
&&Y_tf(r_t)-\int_0^tY_sf'(r_s)\bigg(ab-\Big(a-\sigma\eta-\sigma_Z\int_0^\infty \zeta(e^{-\theta\zeta}-1)\mu_\alpha(d\zeta)\Big)r_s\bigg)ds-\frac{\sigma^2}{2}\int_0^tY_sf''(r_s)r_sds\\
&&\quad\quad\quad-\int_0^t Y_sr_sds\int_0^\infty\Big(f(r_{s-}+\sigma_Z\zeta)-f(r_s)-f'(r_{s-})\sigma_Z\zeta\Big)e^{-\theta\zeta}\mu_\alpha(d\zeta), \quad t\geq 0
\eeqnn
is a local martingale,  which implies that under $\mathbb Q$, $r$ is an $\alpha$-CIR type process with the parameters $(a',b',\sigma',\sigma_Z',\mu_\alpha')$.\qed

\begin{remark}Usually we choose $\eta$ and $\theta$ such that $a'>0$. When $\theta=0$, $\mu_\alpha'$ coincides with $\mu_\alpha$ given in \eqref{Levymeasure}, so that  an $\alpha$-CIR process will remain in the same class under an equivalent change of probability measure. When $\theta>0$, the $\alpha$-CIR process becomes an $\alpha$-CIR type process driven by a tempered stable process under the change of probability measure. In this case, we can apply the following result on the general CBI processes to compute the bond prices. 
\end{remark}


As highlighted by Filipovi\'c \cite{F01,F02}, a large class of bond options admits a nice expression via the exponential affine transformation, see  \cite[Theorem 10.5]{F02} and \cite[Section 6]{F01}. 
The next proposition gives a general result about the joint Laplace transform of a CBI process and its integrated process, which will be useful for the bond pricing. 

\begin{proposition}\label{prop 1}Let $X$ be a CBI $(\Psi,\Phi)$ process given by \eqref{laplace} with $X_0=x$. For non-negative real numbers $\xi$ and $\theta$, we have
\beqlb\label{non consertative}
{\mathbb E}_x\Big[e^{-\xi X_t-\theta\int_0^t X_s ds}\Big]
=\exp\Big(-xv(t,\xi,\theta)-\int_0^t \Phi\big(v(s,\xi,\theta)\big)ds\Big),
\eeqlb
where $v(t,\xi,\theta)$ is the unique solution of
 \beqlb\label{ODE1}
\frac{\partial v(t,\xi,\theta)}{\partial t}=-\Psi(v(t,\xi,\theta))+\theta, \quad v(0,\xi,\theta)=\xi.
 \eeqlb
\end{proposition}

\proof For any function $f\in C^2(\mathbb{R_+})$, $(f(X_t)-f(x)-\int_0^t \mathcal{L}f(X_s) ds,\;  t\geq 0)$ is a local martingale, where 
the operator $\mathcal{L}$ given by (\ref{generator1}).
Then\beqnn
f(X_t)e^{-\theta\int_0^t X_s ds}-f(x)-\int_0^t {e^{-\theta\int_0^tX_sds}}\big(\mathcal{L} f(X_s)-\theta X_sf(X_s)\big)ds, \quad t\geq 0
\eeqnn
is also a local martingale. 
Denote the Feynman-Kac semigroup by $Q_t$ and the corresponding process by $\overline{X}$ as follows:
\beqnn
Q_tf(x)=\overline{\mathbb{E}}_x\left[f\left(\overline{X}_t\right)\right]:=\mathbb{E}_x\Big[f(X_t)e^{-\theta\int_0^tX_sds}\Big].
\eeqnn
Then $\overline{X}$ is a non-conservative 
(in other words the process may explode in finite time)
CBI process with generator 
defined by $\mathcal{A}f(x)=\mathcal{L}f(x)-\theta xf(x)$  
by  \cite[Theorem 1.1]{KaW71} and implies 
in addition that  
 \beqnn
\overline{\mathbb{E}}_x[e^{-\xi \overline{X}_t}]=\exp\Big(-xv(t,\xi,\theta)-\int_0^t\Phi(v(s,\xi,\theta))ds\Big),
\eeqnn
where $v(\cdot,\xi,\theta)$ is the unique solutions of 
 \beqnn\label{ODE1-bis}
\frac{\partial v(t,\xi,\theta)}{\partial t}=-\Psi(v(t,\xi,\theta))+\theta, \quad v(0,\xi,\theta)=\xi.
 \eeqnn
The proposition is thus proved.
\qed


The above proposition allows to compute the zero-coupon bond price with a short rate driven by $\alpha$-CIR model, under an equivalent  probability measure. In the following, we give the zero-coupon price when the short rate $r$ satisfies the $\alpha$-CIR model of parameter $(a,b,\sigma,\sigma_Z,\alpha)$ under the equivalent risk-neutral probability $\mathbb Q$, and analyze its decreasing property with respect to $\alpha$.  Recall that the value of a zero-coupon bond of maturity $T$ at time  $t\leq T$ is given by 
\begin{equation}\label{defaultable bond zero recovery}
B(t,T)=\mathbb E^{\mathbb Q}\Big[\exp\Big(-\int_t^Tr_sds\Big)\,|\,\mathcal F_t\Big].
\end{equation}
For the sake of simplicity,  we will use the notation 
$\mathbb E$ in place of $\mathbb E^{\mathbb Q}$.

\begin{proposition}\label{prop:bond price} Let the short rate $r$ be given by the $\alpha$-CIR model \eqref{lambda-root} under the probability measure $\mathbb Q$.
Then the zero-coupon bond price is given by 
\begin{equation}\label{bond formula}B(t,T)=\exp\Big(-r_t v(T-t)- a b\int_0^{T-t}v(s)ds\Big),\end{equation}
where $v(s)$ is the unique solution of the equation
 \beqlb\label{ODE2}
\frac{\partial v(t)}{\partial t}=1-\Psi_\alpha(v(t)), \quad v(0)=0,
 \eeqlb
with $\Psi_\alpha(q)=a q+\frac{\sigma^2}{2}q^2-\frac{\sigma_Z^\alpha}{\cos(\pi\alpha/2)}q^\alpha$ as in \eqref{equ: Psi SCIR}. Moreover, we have 
\begin{equation}\label{function-f}
v(t)=f^{-1}(t)\, \text{ where }\, f(t) = \int_0^t  \frac{dx}{1-\Psi_\alpha(x)}
\end{equation}

\end{proposition}
\proof  Applying Propositions \ref{prop 1} with $\xi= 0$ and $\theta=1$, we have 
\[\mathbb E \left[\left. e^{-\int_t^Tr_sds}\right|\mathcal F_t\right]
=\exp\Big(-r_t v(T-t)- a b\int_0^{T-t}v(s)ds\Big),\]
where $v(t)$ is the unique solution of (\ref{ODE2})
with $\Psi_\alpha$ given in \eqref{equ: Psi SCIR}. Since $\Psi_\alpha(\cdot)$ is a nonnegative, increasing and convex function,
the equation $\Psi_\alpha(x)=1$ has the unique solution denoted by $x_0$. For $0\leq x<x_0$, $1-\Psi_\alpha(x)>0$. Note that $f(u)$ is strictly increasing in $u\in[0,x_0)$ and 
$f(u)\rightarrow\infty$ as $u\rightarrow x_0$.  It follows from (\ref{ODE2})
that
 \beqnn
 \int_0^{v(t)}  \frac{dv}{1-\Psi_\alpha(v)}=t.
 \eeqnn
Let $t$ tend to infinity on both sides of the above equality.  Then $v(t)\rightarrow x_0$ as $t\rightarrow\infty$ and $v(t)<x_0$ for any $t\geq0$.
Also by (\ref{ODE2}), $v(t)$ is strictly increasing. So one has $v(t)=f^{-1}(t)$.\qed

\begin{proposition}\label{decrasing-bond} The function $v$ is increasing with respect to $\alpha\in(1,2]$. In particular, the bond price $B(0,T)$ is decreasing with respect to $\alpha$.
\end{proposition}
\proof
 We write the function $v$  as $v(t,\alpha)$ to emphasize the dependence on the parameter $\alpha$. 
Since $1-\Psi_\alpha(u)$ is a decreasing concave function of $u$ and $\Psi_\alpha(0)=0$, there is a unique positive solution, denoted by $v^*(\alpha)$, to the equation $1-\Psi_\alpha(u)=0$. It is not hard to see that $0\leq v(s,\alpha)<v^*(\alpha)$ and $\lim_{t\rightarrow\infty}v(s,\alpha)=v^*(\alpha)$. Moreover, from the relation $1-\Psi_\alpha(v^*(\alpha))=0$ we obtain that $(\sigma_Zv^*(\alpha))^\alpha\leq-\cos(\pi\alpha/2)\leq 1$ and hence $\sigma_Zv^*(\alpha)\leq 1$.

For any $t\in\mathbb R_+$, one has
\[ t=\int_0^{v(t,\alpha)}\frac{dx}{1-\Psi_\alpha(x)}.\] Taking the derivative with respect to $\alpha$, we obtain
\[\frac{1}{1-\Psi_\alpha(v(t,\alpha))}\cdot \frac{\partial v}{\partial \alpha}(t,\alpha)+\int_0^{v(t,\alpha)}\frac{1}{(1-\Psi_\alpha(x))^2}\cdot\frac{\partial\Psi_\alpha}{\partial\alpha}(x)dx=0.\]
Note that by \eqref{equ: Psi SCIR},
\[\frac{\partial\Psi_\alpha}{\partial\alpha}(x)=-\frac{\sin(\pi\alpha/2)}{
\cos^2(\pi\alpha/2)}\Big(\frac{\pi}{2}\Big)(\sigma_Zx)^{\alpha}-\frac{(\sigma_Zx)^\alpha}{\cos(\pi\alpha/2)}\ln(\sigma_Zx)\leq 0\]
on $x\in(0,v^*(\alpha)]$ since $\sigma_Zv^*(\alpha)\leq 1$ and $\cos(\pi\alpha/2)<0$. Therefore we obtain $\partial v/\partial\alpha\geq 0$, namely the function $v$ is increasing with respect to $\alpha$. In particular, the bond price $B(0,T)$ is a decreasing function of $\alpha$. \qed 
 
Proposition \ref{decrasing-bond} shows that the $\alpha$-CIR model with $\alpha<2$ permits to capture the low interest rate behavior from the point of view of bond pricing.
This result is  surprising at first sight since the parameter $\alpha$ is an inverse measure of heaviness of distribution tails, more as $\alpha$ close to $1$, more likely that the large jumps appear (see also Section \ref{section-jumps}). In addition, in the $\alpha$-CIR model, $\alpha$ coincides with the so-called generalized Blumenthal-Getoor index which is defined as 
 $ \inf\{\beta>0: \sum_{0\leq s\leq T}\Delta r_s^{\beta}<\infty,\ a.s.\}$ with $\Delta r_s:=r_s-r_{s-}$ and $T$ a horizon time (see e.g. \cite{AJ09}) and is often used to measure the activity of the small jumps in a semimartingale. 
In fact, when $\mu_\alpha(du)$ is defined by (\ref{Levymeasure}), this index is reduced to
$ \inf\big\{\beta>0:\int_0^Tr_{s}ds \int_0^1u^\beta \mu_\alpha(du)<\infty,\ a.s.\big\}$
and thus is equal to $\alpha$.  The index $\alpha\in(1,2)$ shows that  the
jumps are of infinite variation. 
 
%

\subsection{Application to bond derivatives}

We now consider bond derivatives. The $\alpha$-CIR model framework allows to obtain closed-form formulae for a large class of derivatives as we show below by the example of path-dependent option. Denote the zero-coupon bond yield of constant maturity $\kappa$ at time $t$ by $Y(t,t+\kappa)$.  It follows from Proposition \ref{prop:bond price} that  
 \begin{equation}\label{bond yield}
 Y(t,t+\kappa)=-\frac{1}{\kappa}\ln B(t,t+\kappa)=\frac{1}{\kappa}\Big(r_tf^{-1}(\kappa)+ab\int_0^{\kappa} f^{-1}(s) ds\Big).
 \end{equation}
Let us consider a European Put option of maturity $T$ and strike $K$, which is written on the running minimum of the bond yield. The price is given by
\begin{equation}\label{put P price}
P\Big(\inf_{u\in[0,T]}Y(u,u+\kappa),0,T,K\Big) := \mathbb E\Big[e^{-\int_0^T r_sds}\Big(K-\inf_{u\in
[0,T]}Y(u,u+\kappa)\Big)_+\Big]
\end{equation}
We define the Laplace transform with respect to the maturity of the above functional. For $\theta>0$, let 
\begin{equation}\label{laplace put option}
L_\theta\left(0,\kappa, {K};r_0\right)=\int_0^\infty e^{-\theta T}
P\Big(\inf_{u\in[0,T]}Y(u,u+\kappa),0,T,K\Big)dT.
\end{equation}
The following result gives a closed-form expression of this Laplace transform.

\begin{proposition} Let $r$ be an $\alpha$-CIR$(a,b,\sigma,\sigma_Z,\alpha)$ process with initial value $r_0>0$. Then
\beqlb\label{closed form of LS}
L_\theta\left(0,\kappa,{K};r_0\right)=\frac{f^{-1}(\kappa)}{\kappa}\int_0^{\overline{K}}\frac{H_\varepsilon(\theta,r_0)}{H_\varepsilon(\theta,y)}
M(\theta,y) dy,
\eeqlb
where the function $f^{-1}$ is defined in (\ref{function-f}), $ \overline{K}=\left(\kappa K-ab\int_0^\kappa f^{-1}(s)ds\right)/f^{-1}(\kappa)$,
\begin{equation}\label{H}H_{\varepsilon}(\theta,x)=\int_{q_1}^\infty\frac{e^{-xz}}{\Psi_\alpha(z)-1}\exp\Big(\int_{q_1+\varepsilon}^z\frac{abu+\theta}{\Psi_\alpha(u)-1}du\Big)dz,
\end{equation}
with $q_1$ given by 
$\Psi_\alpha(q_1)=1$ and $\varepsilon$ is an arbitrary positive number, and 

\[M(\theta,y)=
\int_0^\infty e^{-\theta u} B_y
(0,u)du
\]
with $B_y(0,u)$ being the zero-coupon bond price given by \eqref{bond formula} with initial short rate $y$.
\end{proposition}


\begin{remark} We note that $H_{\varepsilon}(\theta,x)$ is well defined. Indeed, $\frac{abu+\theta}{\Psi_\alpha(u)-1}\rightarrow0$ as $u\rightarrow\infty$. Then $\frac{1}{z}\int_{q_1+\varepsilon}^z\frac{abu+\theta}
{\Psi_\alpha(u)-1}du\rightarrow0$ as $z\rightarrow\infty$, which
implies $\int_{q_1+\varepsilon}^\infty \frac{dz}{\Psi_\alpha(z)-1}\exp(-yz+\int_{q_1+\varepsilon}^{z}\frac{abu+\theta}{\Psi_\alpha(u)-1}du)<\infty$. In addition, as $z\rightarrow q_1$,  we have
\beqnn
\int_{q_1}^{q_1+\varepsilon}\frac{dz}{\Psi_\alpha(z)-1}\exp\Big(-yz+\int_{q_1+\varepsilon}^{z}\frac{abu+\theta}{\Psi_\alpha(u)-1} du\Big)\leq
\int_{q_1}^{q_1+\varepsilon}\frac{dz}{\Psi_\alpha(z)-1}\exp \Big(\int_{q_1+\varepsilon}^{z}\frac{\theta}{\Psi_\alpha(u)-1}du\Big)<\infty.\eeqnn
In fact, consider $\theta>0$, a primitive function of the integrand on the right hand side is $z\mapsto\frac{1}{\theta}
\exp\big(-\theta\int_z^{q_1+\varepsilon}\frac{1}{\Psi_\alpha(u)-1}du\big)$, which takes finite value at $q_1$.\end{remark}


\begin{proof} We first rewrite the payoff \eqref{put P price} of the Put option as 
 \beqnn
P\left(\inf_{u\in[0,T]}Y(u,u+\kappa),0,T,K\right)
\ar=\ar
\mathbb{E}\left[e^{-\int_0^T r_s ds}\Big( K-\frac{1}{\kappa}\Big[ f^{-1}(\kappa)\inf_{u\in
[0,T]}r_u+ab\int_0^\kappa f^{-1}(s)ds\Big]\Big)_+\right] \\
\ar=\ar\frac{f^{-1}(\kappa)}{\kappa}\mathbb{E}
\left[e^{-\int_0^T r_sds}\Big( \frac{\kappa K-ab\int_0^\kappa f^{-1}(s)ds}{f^{-1}(\kappa)}-\inf_{u\in
[0,T]}r_u\Big)_+\right],
\eeqnn
which corresponds to another Put option written on the running minimum of the spot rate itself with different nominal ${f^{-1}(\kappa)}/{\kappa}$ and strike $\overline K$, i.e., 
 \beqnn
P\Big(\inf_{u\in[0,T]}Y(u,u+\kappa),0,T,K\Big)=\frac{f^{-1}(\kappa)}{\kappa}
P\Big(\inf_{u\in[0,T]}r_u,0,T,\overline{K}\Big).
\eeqnn
Then the Laplace transform \eqref{laplace put option} becomes
\beqlb
L_\theta\left(0,\kappa,{K};r_0\right)=\frac{f^{-1}(\kappa)}{\kappa}\int_0^\infty e^{-\theta T}
P\Big(\inf_{u\in[0,T]}r_u,0,T,\overline{K}\Big)dT.
\eeqlb
Note that
\beqnn
\big(\overline{K}-\inf_{u\in
[0,T]}r_u\big)_+=\int_0^{\overline{K}} 1_{\{\inf_{u\in[0,T]}r_u<y\}}dy,
\eeqnn
hence we have
\beqnn
L_\theta\left(0,\kappa,{K};r_0\right)
\ar=\ar \frac{f^{-1}(\kappa)}{\kappa}\mathbb{E} \Big[  \int_0^{\overline{K}}dy \int_0^\infty dT
\exp\Big(-\theta T-\int_0^T r_sds\Big)1_{\{\inf_{u\in[0,T]}r_u<y\}}
\Big]\\
\ar=\ar \frac{f^{-1}(\kappa)}{\kappa}
\mathbb{E}\Big[\int_0^{\overline{K}} dy\int_{\varTheta_y}^\infty dT
\exp\big(-\theta T-\int_0^T r_sds\big)\Big]\\
\ar=\ar \frac{f^{-1}(\kappa)}{\kappa}
\mathbb{E}\Big[ \int_0^{\overline{K}}dy \int_{\varTheta_y}^\infty dT
\exp\big(-\theta(T-\varTheta_y)-\theta\varTheta_y-\int_0^{\varTheta_y}r_sds-\int_{\varTheta_y}^T r_sds\big)
\Big]\eeqnn where $\varTheta_y$ denotes the first entrance time of $r$ in $[0,y]$ with $y<r_0$ , i.e. $\varTheta_y:=\inf\{t>0: r_t\leq y\}$.
By Duhalde, Foucart and Ma \cite[Theorem 1]{DFM14}, we have \begin{equation} \label{hitting time}
\mathbb{E}
\Big[\exp\Big(-\theta \varTheta_{y}-\int_0^{\varTheta_y}r_t \ddr t\Big)\Big]=\frac{H_\varepsilon(\theta,r_0)}{H_\varepsilon(\theta,y)}
\end{equation}
where the function $H_{\varepsilon}(\theta,x)$ defined in \eqref{H} is a decreasing $C^{2}_0$ function   on $x\in (0,\infty)$ for $\theta>0$.
By using the strong Markov property of $r_t$ on the stopping time $\varTheta_y$,
\beqnn
L_\theta\left(0,\kappa, {K};r_0\right)
\ar=\ar \frac{f^{-1}(\kappa)}{\kappa}
\int_0^{\overline{K}}  \mathbb{E}\Big[  \exp\big(
-\theta\varTheta_y-\int_0^{\varTheta_y}r_sds\big)\Big]\int_0^\infty e^{-\theta t}\mathbb{E}_y\Big[\exp\big(-\int_0^t r_s ds\big)\Big]dt\, dy.
\eeqnn
Note that $B_y(0, t)=\mathbb{E}_y[\exp(-\int_0^t r_s ds)]$, thus we obtain (\ref{closed form of LS}). \qed
\end{proof}



\section{Analysis of jumps}\label{section-jumps}

This section is focused on the jump part of the short interest rate $r$. In particular, we are interested in the large jumps which capture the significant changes in the interest rate dynamics and may imply the downgrade credit risk. 


Similar as in Section \ref{sec:Model}, we fix a jump threshold $\overline{y}=\sigma_Z y>0$. Let $J_t^{\overline{y}}$ denote the number of jumps of $r$ with jump size larger than $\overline{y}$ in $[0,t]$, i.e. 
\begin{equation}\label{jump number J} J_t^{\overline{y}}:=\sum_{0\leq s\leq t}1_{\{\Delta r_s>\overline{y} \}}.\end{equation}
Using the integral representation \eqref{lambda-integral}, we have 
\beqlb\label{JT}
J_t^{\overline{y}}=\int_0^t\int_0^{r_{s-}}\int_{\overline{y}/\sigma_Z}^{\infty} N(ds,du,d\zeta) = 
\int_0^t\int_0^{r_{s-}}\int_{y}^{\infty} N(ds,du,d\zeta),
\eeqlb
where $N$ is the (non-compensated) Poisson random measure corresponding to $\widetilde N$.
Since $\mu_\alpha((0,\infty))=\infty$, we have $\lim_{\overline{y}\rightarrow0}J^{\overline{y}}_t=\infty$, a.s.. In the following, we show that the Laplace transform of this counter process is exponential affine where the exponent  coefficient satisfies  a non-linear ordinary differential equation.



\begin{proposition} \label{Th3.7}Let $r$ be $\alpha$-CIR$(a,b,\sigma,\sigma_Z,\alpha)$ process with initial value $r_0\geq0$. Then for $p\geq0$,
\beqlb\label{lt}
\mathbb{E}\big[e^{-pJ_t^{\overline{y}} }\big]=\exp\Big({-l(p,y,t)r_0-ab\int_0^tl(p,y,s)ds}\Big)
\eeqlb
where $l(p,y,t)$ is the unique solution of the following equation 
 \beqlb\label{ODEl}
 \frac{\partial l(p,y,t)}{\partial t}=\sigma_Z^\alpha\int_y^\infty \big(1-e^{-p-l(p,y,t)\zeta}\big)\mu_\alpha(d\zeta)-\Psi^{(y)}_\alpha(l(p,y,t)),
 \eeqlb
with initial condition $l(p,y,0)=0$ and  
$\Psi^{(y)}_\alpha$ given by (\ref{Psix}).

\end{proposition}
\proof 
We first show that (\ref{ODEl}) has a unique solution. 
Note that $\sigma_Z^\alpha\int_y^\infty\mu_\alpha(d\zeta)-\Psi^{(y)}(q)$ is a decreasing concave function with respect to
$q$ and $\sigma_Z^\alpha\int_y^\infty e^{-p-q\zeta}\mu_\alpha(d\zeta)$ is a decreasing convex function of $q$. Since $p\geq0$, one has 
$\sigma_Z^\alpha\int_y^\infty\mu_\alpha(d\zeta)-\Psi^{(y)}(0)\geq \sigma_Z^\alpha\int_y^\infty e^{-p}\mu_\alpha(d\zeta)$. Moreover, for $q$ large enough, 
$\sigma_Z^\alpha\int_y^\infty\mu(d\zeta)-\Psi^{(y)}(q)<0< \sigma_Z^\alpha\int_y^\infty e^{-p-q\zeta}\mu_\alpha(d\zeta)$. Thus there is the unique positive solution, denoted by $l^*>0$, to the equation
\beqnn
F_y(q):=\sigma_Z^\alpha\int_y^\infty(1-e^{-p-q\zeta})\mu(d\zeta)-\Psi^{(y)}(q)=0.
\eeqnn
One has $F_y(q)>0$ when $0\leq q <l^*$, and $F_y(q)<0$ when $q >l^*$. Moreover
$
\Gamma(l):=\int_0^l\frac{1}{F_y(q)}dq
$
is an increasing function from $[0,l^*)$ to $[0,\infty)$ and its inverse function $l(p,y,\cdot): [0,\infty)\rightarrow[0,l^*)$ exists. 
It is not hard to see that for any $t\geq0$,
$
\int_0^t\frac{1}{F_y(l(p,y,s))}dl(p,y,s)=t,
$
which implies (\ref{ODEl}). Since $F_y(q)$ is locally Lipschitz, the uniqueness follows.

The couple $(J^{\overline{y}},r)$
is a Markov process taking values in $\mathbb{N}\times\mathbb{R}_+$, where $\mathbb{N}:=\{0,1,\cdots\}$. The generator of $(J^{\overline{y}},r)$ acting on a function $f(x,n,t)$ is given by
 \begin{equation}\label{ODEB}
 \begin{split}
 \mathcal{A} f(x,n,t)= \frac{\partial f}{\partial t}&(x,n,t) +a(b-x)\frac{\partial f}{\partial x}(x,n,t)
 +\frac{1}{2}\sigma^2x\frac{\partial ^2f}{\partial x^2}(x,n,t) \\
& +\sigma_Z^\alpha x\int_0^y\big(f(x+\zeta,n,t)-f(x,n,t)-  \zeta \frac{\partial f}{\partial x} (x,n,t)
\big)\mu_\alpha(d\zeta)\\
&
+\sigma_Z^\alpha x\int_y^{\infty} \big(f(x+\zeta,n+1,t)-f(x,n,t)-\zeta\frac{\partial f}{\partial x} (x,n,t)\big)\mu_\alpha(d\zeta),
\end{split}
 \end{equation}
where $f(x,n,t)$ is differentiable with respect to $t$ and twice differentiable with respect to $x$,  and the measure 
$\mu_\alpha(d\zeta)$ is defined by (\ref{Levymeasure}). 
Let $p$ and $\theta$ be non-negative numbers, and $T\geq 0$ be a time horizon. We consider the integral-differential equation $\mathcal Af=0$ with boundary condition $f(x,n,T)=\exp(-pn-\theta x)$ and look for  a solution of the form
\[f(x,n,t)=\exp\big(C_0(t)-C_1(t)n-C_2(t)x\big), \quad t\in[0,T].\]
Then the equation $\mathcal Af=0$ reduces to the following system of ordinary differential equations
\begin{equation}\label{equ:C012}\begin{cases}C_0'(t)=abC_2(t),\\  C_1'(t)=0,\\C_2'(t)=\Psi_\alpha^{(y)}(C_2(t))+
\sigma_Z^\alpha\int_y^\infty(e^{-C_2(t)\zeta-C_1(t)}-1)\mu_\alpha(d\zeta).
\end{cases}\end{equation}
Moreover, the boundary condition $f(x,n,T)=\exp(-pn-\theta x)$ reads 
$(C_0(T),C_1(T),C_2(T))=(0,p,\theta)$. 
In particular, one has $C_1(t)=p$ on $t\in[0,T]$. Moreover, the functions $C_0$ and $C_2$ are also uniquely determined by the equation system \eqref{equ:C012} and the boundary condition. 
Notably one has $C_0(t)=-ab\int_t^TC_2(s)ds$.
Since $\mathcal A$ is the generator of the Markov process $(J^{\overline y},r)$, one has
\[\mathbb E[e^{-pJ_T^{\overline{y}}-\theta r_T}|\mathcal F_t]=f(r_t,J_t^{\overline{y}},t)=\exp\Big(-ab\int_t^TC_2(s)ds-pJ_t^{\overline y}-C_2(t)r_t\Big),\]
where $C_2$ is the solution of the following ordinary differential equation with boundary condition
\[C_2'(t)=\Psi_\alpha^{(y)}(C_2(t))+
\sigma_Z^\alpha\int_y^\infty(e^{-C_2(t)\zeta-p}-1)\mu_\alpha(d\zeta),\quad C_2(T)=\theta.\]
The particular case where $\theta=0$ and $t=0$ leads to
\begin{equation}\label{equ:CT}\mathbb E[e^{-pJ_T^{\overline{y}}}]=\exp\Big(-ab\int_0^T C(T,p,y,s)ds-C(T,p,y,0)r_0\Big),\end{equation}
with $C(T,p,y,\cdot)$ being the solution of
\[\frac{\partial{C(T,p,y,t)}}{\partial t}=\Psi_\alpha^{(y)}(C(T,p,y,t))+
\sigma_Z^\alpha\int_y^\infty(e^{-C(T,p,y,t)\zeta-p}-1)\mu_\alpha(d\zeta),\quad C(T,p,y,T)=0.\]
Finally, the comparison between the differential equations \eqref{ODEl} and \eqref{equ:CT} shows that 
 $l(p,y,t)=C(T,p,y,T-t)$ for any $t\leq T$. 
Hence we obtain \eqref{lt}.\qed


Now we consider the first time when the jump size of  the short rate $r$ is larger than $\overline{y}=\sigma_Zy$, i.e.,
\begin{equation}\label{first jump time}\tau_{\overline{y}}=\inf\{t>0: \Delta r_t>\overline{y} \}.\end{equation} 
 We show that this random time also exhibits  an exponential affine cumulative distribution function. 
The following result gives its distribution function as a consequence of the above proposition.

\begin{corollary}\label{distribution of taux}  For any $t\geq 0$, we have
\beqlb\label{maximal jump}
\mathbb P(\tau_{\overline{y}}>t)=\exp \Big(\displaystyle-l(y,t)r_0-ab\int_0^tl(y,s)ds\Big)
\eeqlb
where $l(y,t)$ is the unique solution of the following ODE 
 \beqlb\label{ODEll}
 \frac{dl}{dt}(y,t)=\sigma_Z^\alpha\int_y^\infty\mu_\alpha(d\zeta)-\Psi^{(y)}_\alpha(l(y,t)),
 \eeqlb
with initial condition $l(y,0)=0$ and $\Psi^{(y)}_\alpha$ given by (\ref{Psix}).
\end{corollary}

\proof For $q\geq 0$ one has
\[\sigma_Z^\alpha\int_y^\infty(1-e^{-p-q\zeta})\mu_\alpha(d\zeta)-\Psi^{(y)}(q)\leq \sigma_Z^\alpha\int_y^\infty \mu_\alpha(d\zeta)-aq.\]
By the equation (\ref{ODEl}) in Proposition \ref{Th3.7},  we obtain that 
 \beqlb\label{bound}
 l(p,y,t)\leq\frac{\sigma_Z^\alpha}{a}\left(1-e^{-at}\right)\int_y^\infty\mu_\alpha(d\zeta),
  \eeqlb
 and $l(p,x,t)$ is increasing of $p$. Thus $l(y,t):=\lim_{p\rightarrow\infty}l(p,y,t)$ exists. By (\ref{ODEl}), 
 \beqnn
 l(p,y,t)=\sigma_Z^\alpha\int_0^t \Big(\int_y^\infty(1-e^{-p-l(p,y,s)\zeta})\mu_\alpha(d\zeta)-\Psi^{(y)}_\alpha(l(p,y,s))\Big)ds 
 \eeqnn
 Since $\Psi^{(y)}(q)$ is locally Lipschitz and $e^{-p-l(p,y,s)\zeta}\leq
 e^{-p}$, by taking limit as $p\rightarrow\infty$ on the both sides of the above equation we have 
  \beqnn
 l(y,t)=\int_0^t \Big(\sigma_Z^\alpha\int_y^\infty\mu(d\zeta)-\Psi^{(y)}(l(y,s))\Big)ds,  
  \eeqnn
which implies that the limit function $l$ is the unique solution of the equation (\ref{ODEll}). By Proposition \ref{Th3.7} and (\ref{JT}),
\beqnn
\mathbb P(\tau_{\overline{y}}>t)=\mathbb P(J^{\overline{y}}_t=0)=\lim_{p\rightarrow\infty}\mathbb{E}\big[e^{-pJ_t^{\overline{y}}}\big]=\exp\Big({-l(y,t)r_0-ab\int_0^tl(y,s)ds}\Big).
\eeqnn
The last equality follows from the monotone convergence theorem. \qed

\begin{proposition}\label{pro:expectation duration} We have $\mathbb{P}(\tau_{\overline{y}} <\infty)=1$. Furthermore, denote $F(q):=\sigma_Z^\alpha\int_y^\infty\mu_\alpha(d\zeta)-\Psi^{(y)}_\alpha(q)$, then the equation $F(q)=0$ admits a unique solution $l^*_y$, which identifies with $\lim_{t\rightarrow\infty}l(y,t)$ where  $l$ is given by \eqref{ODEll}. Moreover, one has
\beqlb \label{expectation}
\mathbb{E}\left[\tau_{\overline{y}} \right]=\int_0^{l^*_y}\frac{1}{F(u)}\exp\Big(
-ur_0-\int_0^u\frac{abs}{F(s)}ds\Big)du<\infty.
\eeqlb

\end{proposition}
\proof We note that $F$ is a
decreasing concave function and $F(0)>0$. Hence 
the equation $F(q)=0$ has a unique positive solution $l^*_y>0$. One has $F(q)>0$ when $q\in[0,l^*_y)$. 
By (\ref{ODEll}),
\begin{equation}\label{equ:equationlt}
\int_0^{l(y,t)}\frac{1}{F(q)}dq=t,
\end{equation}
which implies that $0\leq l(y,t)<l^*_y$ for any $t\geq0$. Then $l(y,t)$ is strictly increasing on $t$. Let $t$ tend to infinity in the above equality \eqref{equ:equationlt}, we deduce
that $\lim_{t\rightarrow\infty}l(y,t)=l^*_y>0$. Then  $\int_0^\infty l(y,s)ds=\infty$. By Corollary \ref{distribution of taux}, $\mathbb{P}(\tau_{\overline{y}}=\infty)=0$.
Note that $\mathbb{E}[\tau_{\overline{y}}]=\int_0^{\infty}\mathbb{P}(\tau_{\overline{y}}>t)dt$, so
\beqnn
\mathbb{E}[\tau_{\overline{y}} ]=\int_0^\infty \exp\Big({-l(y,t)r_0-ab\int_0^tl(y,s)ds}\Big)dt=\int_0^{l^*_y}
\frac{1}{F(q)}\exp\Big({-qr_0-\int_0^{q}\frac{abp}{F(p)}dp}\Big)dq,
\eeqnn
where the second equality follows from (\ref{ODEll}) and implies (\ref{expectation}).  Since $F$ is decreasing, $F'(l^*_x)<0$ and then 
by concavity
\beqnn
\frac{1}{F(u)}\exp\Big({-ur_0-\int_0^u\frac{abs}{F(s)}ds}\Big)\,\sim\,
\frac{c}{F^\prime(l^*_y)(u-l^*_y)}\exp\Big({-ur_0-\int_0^u\frac{abs}{F^\prime(l^*_y)(s-l^*_y)}ds}\Big)
\eeqnn
for some constant $c>0$, as $u\rightarrow  l^*_y$. Then $E[\tau_{\overline{y}} ]<\infty$ follows from
 \beqnn
 \int_0^{l^*_y}\frac{1}{F^\prime(l^*_y)(u-l^*_y)}\exp\Big({-ur_0-\int_0^u\frac{abs}{F^\prime(l^*_y)(s-l^*_y)}ds}\Big)du<\infty. \eeqnn\qed

The following result gives an alternative form of Corollary \ref{distribution of taux} and a more intuitive explanation. It shows that  the distribution of the first jump time  $\tau_{\overline{y}}$  can also be given by using the Laplace transform of the
integrated auxiliary process $\widehat{r}^{(y)}$, which is introduced previously in (\ref{rhat1}), computed on the mass of the jump measure larger than $y= \overline{y}/\sigma_Z$. 
In other words,  the probability $\mathbb{P}(\tau_{\overline{y}}>t)$ is equal to a bond price written on the auxiliary rate $\widehat{r}^{(y)}$ remodulated by 
the measure $\mu_\alpha$ restricted on $(y,+\infty)$. 
When $b=0$, it 
recovers a result of He and Li \cite[Theorem 3.2]{HL15}. 

 \begin{proposition}\label{prop-taux-aux}
 Let $\widehat{r}^{(y)}$ be defined by (\ref{rhat1}), then we have
 \beqlb\label{probability rep}
\mathbb{P}(\tau_{\overline{y}}>t)=\mathbb{E}\Big[\exp{\Big\{-\sigma_Z^\alpha\Big(\int_y^\infty\mu_\alpha(d\zeta)\Big)\Big(\int_0^t\widehat{r}^{(y)}_sds\Big)\Big\}}\Big].
\eeqlb
 \end{proposition}

\proof  
As proved in Corollary \ref{pro: auxiliary CBI}, $\widehat{r}^{(y)}$ is a CBI process.
Then, applying Proposition \ref{prop:bond price}, for any $\theta>0$, we have
\beqnn
\mathbb{E}\Big[e^{-\theta\int_0^t\widehat{r}^{(y)}_sds}\Big]=\exp{\big(\widehat{l}(\theta,t)r_0-ab\int_0^t\widehat{l}(\theta,s)ds\big)},
\eeqnn
where $\widehat{l}(\theta,t)$ is the unique solution of 
\beqnn
\frac{d\widehat{l}(\theta,t)}{dt}=\theta-\Psi^{(y)}_\alpha\big(\widehat{l}(\theta,t)\big),
\eeqnn
with $\widehat{l}(\theta,0)=\theta$. Then Corollary \ref{distribution of taux} can be rewritten in the form (\ref{probability rep}).
\qed

Finally, we compare  the behaviors of the  first large jump times in $\alpha$-CIR and locally equivalent LOU models respectively.  

\begin{proposition}\label{proposition 5.5}
 Let $\tau_{\overline{y}}^{\lambda}:=\inf\{t>0: \Delta \lambda_t>\overline{y} \}$ 
 denote the first time when the jump size of a LOU process $\lambda$ is larger than $\overline{y} := \sigma_Z y$, in accord with Definition \ref{def-LevyOU}. 
 Let $\tau^r_{\overline{y}}$ be defined as in Corollary \ref{distribution of taux}. 
 Then we have the two following relations
  
  \beqlb
\label{equ:ptaulamaba}  \mathbb{P}\left(\tau_{\overline{y}}^{\lambda}\leq t\right) \ar=\ar  1- \exp\left(- C_{\alpha} r_0t\; y^{-\alpha} \right)  \\
\label{equ:prt}
\mathbb{P}\left(\tau^r_{\overline{y}}\leq t\right)\ar \leq \ar C_{\alpha} \, y^{-\alpha}
\Big(    \widetilde{b}(\alpha,y) t      +\frac{r_0-\widetilde{b}(\alpha,y)}{\widetilde{a}(\alpha,y)}\left[ 1-   e^{-\widetilde{a}(\alpha,y)t}   \right] 
\Big) \, ,
\eeqlb
where $\widetilde{a}$ and $\widetilde{b}$ are defined by (\ref{ab}) 
and $C_{\alpha} :=  \frac{2}{\pi}\Gamma(\alpha)\sin(\pi\alpha/2)$. 
Moreover, we have the two following asymptotic tail probabilities
of maximal jump as $\overline{y}$ goes to infinity.  
\beqlb
\mathcal{M}_\lambda(t,\overline{y}):=\mathbb{P}\left(\sup_{0\leq s\leq t} \Delta \lambda_s>\overline{y}\right)\ar \sim \ar 
C_{\alpha} \, r_0 \, t\,   (\overline{y})^{-\alpha}  \\ 
\mathcal{M}_r(t,\overline{y}):=\mathbb{P}\left(\sup_{0\leq s\leq t}\Delta r_s>\overline{y}\right)
\ar \sim \ar C_{\alpha} 
\Big(   b t      +\frac{r_0-b}{a}( 1-   e^{-at} ) 
\Big)  (\overline{y})^{-\alpha}.
 \eeqlb
 \end{proposition}
\begin{remark} Before giving the proof of this result, we note that, comparing $\mathcal{M}_\lambda$ and $\mathcal{M}_r$ when  $t$ goes to zero, we have that the two asymptotic tail probabilities coincide. Whereas when $t$ is large enough,  $\mathcal{M}_r$  is approximately proportional to the long term interest rate $b$. 
\end{remark}
 \proof
 By (\ref{Levy-OU-integral-without-compensation}), we have 
$$
  \mathbb{P}\left(\tau_{\overline{y}}^{\lambda}>t\right)=
   \mathbb{P}\left(\int_0^t\int_0^{r_0}\int_{y}^\infty N(ds,du,d\zeta)=0 \right)
$$
 Then the first result \eqref{equ:ptaulamaba} is obtained by a direct integration.
 For the $\alpha$-CIR case, applying Proposition \ref{prop-taux-aux}, we have
  \beqlb\label{tau_y(r)}
    \mathbb{P}\Big(\tau^r_{\overline{y}}>t\Big)  \ar=\ar
\mathbb{E}\Big[\exp{- \Big\{C_\alpha y^{-\alpha}\int_0^t\widehat{r}^{(y)}_sds\Big\}}\Big].
  \eeqlb
Thus we obtain $\mathbb{E}[\widehat{r}^{(y)}_t]=\widetilde{b}(\alpha,y)\Big(1-e^{-\widetilde{a}(\alpha,y)t}\Big)+r_0 e^{-\widetilde{a}(\alpha,y)t}$, by (\ref{tau_y(r)})
we obtain the second result \eqref{equ:prt} by convexity. 

The first asymptotic tail is a direct consequence of the relation $
\mathbb{P}\big(\sup_{0\leq s\leq t} \Delta \lambda_s>\overline{y}\big) =1-
\mathbb{P}\big(\tau_{\overline{y}}^{\lambda} < t\big)$. 
For the asymptotic tail of $r$, by  (\ref{ODEll}), we have that 
 \begin{equation}\label{technical 1}
 \begin{array}{rcl}
 l(y,t)&=&\displaystyle \sigma_Z^\alpha\Big(\int_y^\infty\mu_\alpha(d\zeta)\Big)\Big(\int_0^\infty e^{-a(t-s)}ds\Big)-\sigma_Z^\alpha\Big(\int_y^\infty\zeta\mu_\alpha(d\zeta)\Big)\Big(\int_0^te^{-a(t-s)}l(y,s)ds\Big) \\
&& \displaystyle -\frac{\sigma^2}{2}\int_0^te^{-a(t-s)}l^2(y,s)ds-\int_0^te^{-a(t-s)}\overline{\Psi}^{(y)}_\alpha(l(y,s))ds,
 \end{array}
 \end{equation}
 where $\overline{\Psi}^{(y)}_\alpha(q)=\sigma_Z^\alpha\int_0^y(e^{-q\zeta}-1+q\zeta)\mu_\alpha(d\zeta)$.  This also shows that 
\beqlb\label{technical 2}
 l(y,t)\leq-\frac{\sigma_Z^\alpha}{a\cos(\pi\alpha/2)\alpha\Gamma(-\alpha)}(1-e^{-at})y^{-\alpha}=C_\alpha\frac{\sigma_Z^\alpha}{a}(1-e^{-at})y^{-\alpha}
  \eeqlb
since $-(\alpha\cos(\pi\alpha/2)\Gamma(-\alpha))^{-1}=C_\alpha$.
By (\ref{technical 1}), we also have that 
\beqnn
y^{\alpha}l(y,t)\ar=\ar-\frac{\sigma_Z^\alpha}{\alpha\cos(\pi\alpha/2)\Gamma(-\alpha)}
\int_0^te^{-a(t-s)}ds+\frac{\sigma_Z^\alpha}{(\alpha-1)\cos(\pi\alpha/2)\Gamma(-\alpha)}y^{-1}\int_0^te^{-a(t-s)}l(y,s)ds\\
 \ar\ar-\frac{\sigma^2}{2}\int_0^te^{-a(t-s)}l^2(y,s)y^\alpha ds-\int_0^te^{-a(t-s)}\overline{\Psi}^{(1)}_\alpha(y\, l(y,s))ds.
 \eeqnn
Combing (\ref{technical 2}), we see that   as $y\rightarrow\infty$, 
\beqlb\label{technical 3}
y^{\alpha}l(y,t)\rightarrow-\frac{\sigma_Z^\alpha}{\alpha\cos(\pi\alpha/2)\Gamma(-\alpha)}
\int_0^te^{-a(t-s)}ds=C_\alpha\sigma_Z^\alpha\frac{1-e^{-at}}{a}. 
\eeqlb
 Furthermore this convergence is locally uniform for $t$. By Corollary \ref{distribution of taux},
\beqnn
\mathbb P(\sup_{0\leq s\leq t}\Delta r_s>\overline{y})=\mathbb P(\tau^r_{\overline{y}}\leq t)=1-e^{-l(y,t)r_0-ab\int_0^tl(y,s)ds}\sim
l(y,t)r_0+ab\int_0^tl(y,s)ds.
\eeqnn
We have the tail of the jump of $r$ by (\ref{technical 3}).
\qed

\begin{remark}
Consider $\widehat{r}^{(y)}$ defined by (\ref{rhat1}). We have noted that  for $0<t<
 \tau_{\overline{y}}^{r}$, $r_t=\widehat{r}^{(y)}_t$. Then for any fixed $T$,
  \[
  \sup_{0\leq t\leq T}\Big|E\Big[\exp\Big\{-\int_0^t r_s ds\Big\}\Big]-E\Big[\exp\Big\{-\int_0^t\widehat{r}^{(y)}_sds\Big\}\Big]\Big|
  \leq 2\mathbb P(\tau_{\overline{y}}^r\leq T)= \mathbb P(\sup_{0\leq s\leq T}\Delta r_s>\overline{y})
    \]
By Proposition \ref{proposition 5.5}, one has $\mathbb P(\sup_{0\leq s\leq T}\Delta r_s>\overline{y})  \sim C(T)y^{-\alpha}
$, where $C(T)$ is a constant depending on $T$. This means that as $y\rightarrow\infty$, $r$ can be approximated by 
$\widehat{r}^{(y)}$ with rate $y^{-\alpha}$.  In the approximation sense, we see that the role  of big jumps which leads to  the additional  negative 
drift term shown in  (\ref{rhat1}) and forces the interest rate at a low level as $\alpha$ decreases to 1.
\end{remark}

\section{Numerical illustration}\label{sec:numerics}
In this section, we present numerical examples to illustrate the results obtained in previous sections. We are particularly interested in the role of the parameter $\alpha$.

In the first example, we present in Figure \ref{fig:zt} a trajectory of  the $\alpha$-stable L\'evy process $Z$ for three different values of $\alpha$: 2, 1.5 and 1.2 respectively. The other parameters are fixed to be $a=0.1$, $b=0.3$, $\sigma=0.1$, $\sigma_Z=0.3$ and $r_0=0.1$. We see that  smaller values of $\alpha$ imply larger jumps and deeper negative drift between the jumps  in the  process $Z$.    We then illustrate in Figure \ref{fig:rt} 
the $\alpha$-CIR process for the short interest rate $r$ described in Definition \ref{def-SDE-root}, by using the same trajectory of $Z$ as in Figure \ref{fig:zt}. 
We observe that since the jumps are related to the actual level of the interest rate, the smaller values of $\alpha$ correspond to a persistency of low interest rate in Figure \ref{fig:rt}. 

\begin{figure}[H]
\caption{The $\alpha$-stable L\'evy process $Z$ for different values of $\alpha$.}
\label{fig:zt}
\begin{center}
\epsfig{file=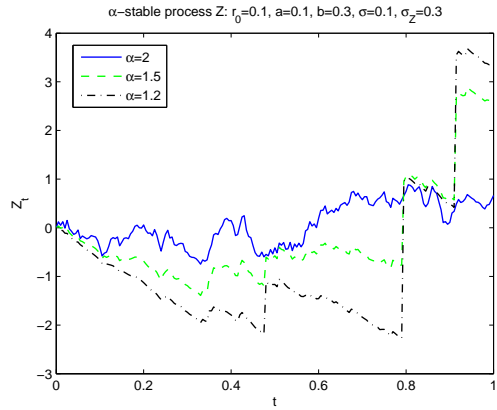,height=7.5cm,angle=270}
\end{center}
\end{figure}

\begin{figure}[H]
\caption{Short interest rates $r$ by the $\alpha$-CIR model with the same $Z$.}
\label{fig:rt}
\begin{center}
\epsfig{file=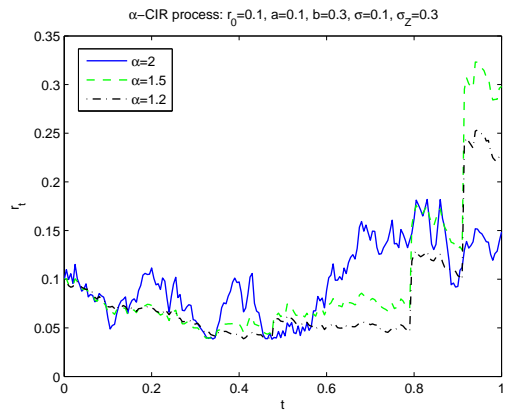,height=7.5cm,angle=270}
\end{center}
\end{figure}

In the second example, we show  by Figure \ref{fig:bond price} the sovereign bonds price $B(0,T)$ given in Proposition \ref{prop:bond price}. The parameters are $a=0.1$, $b=0.3$, $\sigma=0.1$, $\sigma_Z=0.3$ and $r_0=0.05$. Besides the three values of $\alpha$: 2, 1.5 and 1.2, we also consider the bond price in the classical CIR model (when $\sigma_Z=0$). It is interesting to note, as already shown in Proposition \ref{decrasing-bond},  that for a fixed maturity, the bond prices are decreasing with respect to the value of $\alpha$, with the lowest price in the CIR model. This observation means that smaller $\alpha$ corresponds,  in expectation sense, to a lower interest rate phenomenon, even though this case also implies larger positive jumps in the short rate (as in the next figures). 
\begin{figure}[H]
\caption{Sovereign bond prices $B(0,T)$}
\label{fig:bond price}
\begin{center}
\epsfig{file=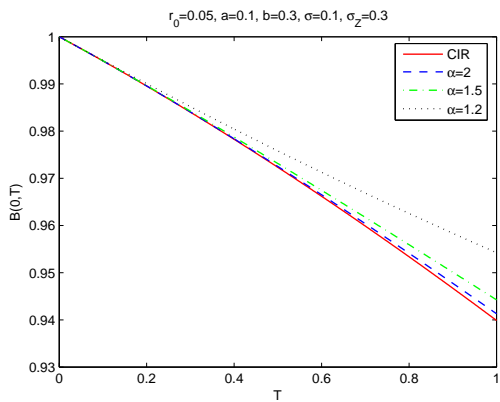,height=7.5cm,angle=270}
\end{center}
\end{figure}


Finally, we illustrate the behaviors of the first large jump $\tau_{\overline y}$ (as in \eqref{first jump time}) that the short rate process exceeding $\overline y$. The parameters are  $a=0.1$, $b=0.1$, $\sigma=0.1$, $\sigma_Z=0.1$, $r_0=0.2$ and $y=0.1$. Figure \ref{fig:proba duration} shows the probability function $\mathbb P(\tau_{\overline y}>t)$, given by Corollary \ref{distribution of taux}, for different values of $\alpha$.  We see that this probability converges to $0$ very quickly for smaller values of $\alpha$, and with a much longer time for  large values of $\alpha$. 
Figure \ref{fig:expectation duration} illustrates the expectation of $\tau_{\overline y}$ given by Proposition \ref{pro:expectation duration}, as a function of $\alpha$. The expected jump time is increasing with $\alpha$, which means that  for a smaller $\alpha$, the first large jump is likely to occur sooner.  These two tests show that the $\alpha$-CIR model with $\alpha<2$ allows to describe the large jumps in the interest rate. 
\begin{figure}[H]
\caption{Probability function $\mathbb P(\tau_{\overline y}>t)$ for the first large jump exceeding $\overline y$}
\label{fig:proba duration}
\begin{center}
\epsfig{file=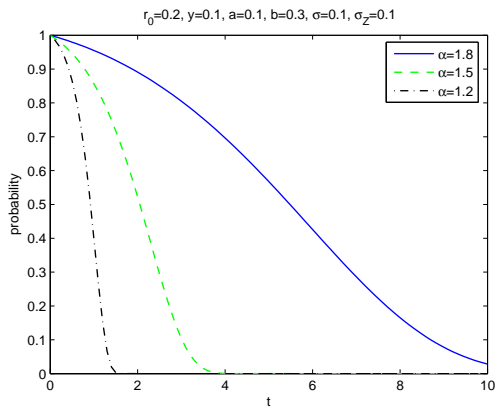,height=7.5cm,angle=270}
\end{center}
\end{figure}

\begin{figure}[H]
\caption{Expectation of the duration time $\tau_{\overline y}$ for the first large jump exceeding $\overline y$}
\label{fig:expectation duration}
\begin{center}
\epsfig{file=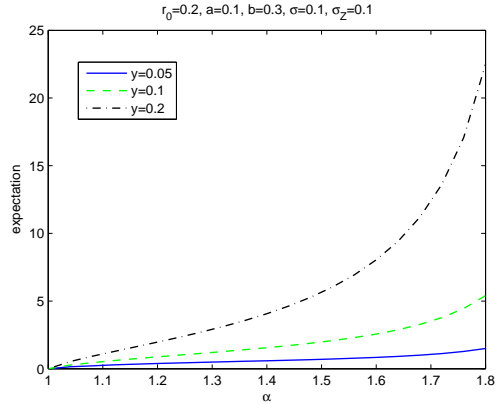,height=7.5cm,angle=270}
\end{center}
\end{figure}

\section{Conclusion}\label{sec:conclusion}

The objective of this paper is to introduce  the $\alpha$-CIR short interest rate model, which is an extension of the standard CIR model by adding, besides the Brownian motion, a spectrally positive $\alpha$-stable L\'evy process and preserving the branching property.

Our main financial contribution is to describe in a parsimonious framework a number of well-established and seemingly puzzling 
facts observed in the current sovereign bond markets. In particular, we reconcile in this relatively simple model the presence of significant variations of the interest rates together with the actual persistency of very low interest rates. 
Moreover, the evolution of the interest rate in our model exhibits the clustering or self-exciting properties which are recently highlighted in stochastic modelling especially in finance. 
 
An interesting financial result is that the bond price increases with the tail fatness of the jump process which is
counter-intuitive and opens the discussion of the consequence on the risk analysis. In particular, our model forecasts 
that the persistency of very low interest rates is accentuated by this tail fatness and this persistency is 
statistically broken by the arrival of the first large jump whose expected arrival probability decreases with the rate itself.

The main mathematical contribution is the introduction of a more general integral representation of the $\alpha$-CIR model by using random fields. This integral representation contributes largely to simplify the mathematical
proofs and helps to establish a link of our model to the CBI processes, and then to the affine interest rate models. We also characterize, using this representation, the law of the frequency of large jumps and the law of the first one. 
 
From computational point of view, we show that our model admits closed-form formulae up to numerical integrations 
for a large class of relevant quantities, for instance for the bond prices and derivatives, and also for the law and the expectation of the first large jump. The perspective of further research work consists of empirical and statistical analysis of the $\alpha$-CIR interest rate model. The integral representation may also open a  range of extensions for other financial modelling. 


\appendix

\section{Constructive proof of Proposition \ref{pro: SCIR CBI}}

\proof {\bf Step 1: branching without immigration.}
 Consider a special case of \eqref{lambda-integral} with $b=0$ and we call it a CB process, i.e. a continuous state branching process 
{\bf without immigration},
 \beqnn
 r_t^x= x -a\int_0^t r_s^x ds + \sigma \int_0^t \int_0^{r_s^x} W(ds,du)
+ \sigma_Z \int_0^t \int_0^{r_{s-}^x} \int_{\mathbb{R}^+} \zeta \widetilde{N}
(ds,du,d\zeta),
 \eeqnn
with initial value $x\geq0$. By the proof of Proposition \ref{prop:decomposition}, $r_t^x$ is increasing of $x$. Furthermore, for $x\geq y$,
$r_t^x-r_t^y$ is independent of $r_t^y$ and have the same distribution of 
$r_t^{x-y}$. Then for any $t$, $\{r_t^x:x\geq0\}$ is a L\'{e}vy
subordinator. The L\'{e}vy-Khintchine Formula implies that
\beqnn
\mathbb{E}[e^{-pr_t^x}]=e^{-xv(t,p)}
\eeqnn
for some L\'{e}vy exponent $v(t,p)$ and $v(0,p)=p$. Since  $\{r_s^x:s\geq0\}$  is the unique strong solution of 
the above equation, it is a Markov process, i.e., $E[\exp({-pr_{t+s}^x})|\mathcal{F}_s]=\exp({r_s^xv(t,p)})$, which implies
that  $v(s,v(t,p))=v(s+t,p)$.
Apply It\^{o}'s formula to  $\exp({-pr_s^x})$ and take the expectation, 
\beqnn
e^{-xv(s,p)}-e^{-px}=\Psi(p)\int_0^s\mathbb{E}[e^{-pr_u^x}r_u^x]du.
\eeqnn
Fix $t$. Replace $p$ by $v(t,p)$ in the above equation,
\beqnn
e^{-xv(s+t,p)}-e^{-xv(t,p)}=\Psi(v(t,p))\int_0^s \mathbb{E}\left[e^{-v(t,p)r_u^x}r_u^x\right]du.
\eeqnn
Differentiating both sides of the equation w.r.t $s$ at $s=0$, we have that $\frac{\partial{v(t,p)}}{\partial{t}}=-\Psi(v(t,p))$.

{\bf Step 2: introduction of the auxiliary jump process.} Let $c>0$ and let $J_t$ 
be a Poisson process with parameter $\lambda_J>0$ independent of $(W,N)$.  Then we define the process $Y$ with initial value 
$x$ by
 \beqlb\label{poisson immigration}
 Y^x_t = x +cJ_t-a\int_0^t Y^x_s ds + \sigma \int_0^t \int_0^{Y^x_s} W(ds,du)
+ \sigma_Z \int_0^t \int_0^{Y^x_{s-}} \int_{\mathbb{R}^+} \zeta \widetilde{N}
(ds,du,d\zeta).
 \eeqlb
Let us assume the following:
(a) At $t=0$ there is one individual with mass $x$. It evolves and gives mass $r_t^x$ at time $t>0$. 
(b) Immigrants with each mass $c$ arrive according to the Poisson process $J_t$. 
The arrival times of $J_t$ is denoted by $0\leq\tau_1\leq\cdots\leq\tau_n\leq\cdots$.
If one immigrant arrives at time $\tau_k$, it gives mass $r^{(k)}_{t-\tau_k}$ at time
$t$, where $r^{(k)}_\cdot$ is an independent copy of $r_\cdot(c)$.  Then we have that 
 \beqlb\label{excursion}
 Y^x_t=r_t^x+\sum_{k=1}^{J_t}r^{(k)}(t-\tau_k).
 \eeqlb
 We now define a Picard sequence $\widehat{r}^{(k)}_t$ by the first step $\widehat{r}^{(0)}_t=r_t^x$ 
 and the relation between $\widehat{r}^{(k-1)}_t$ and $\widehat{r}^{(k)}_t$ defined as
 \beqnn
  \widehat{r}^{(k)}_t = \widehat{r}^{(k-1)}_{\tau_k}+c -a\int_0^t \widehat{r}^{(k)}_s ds + \sigma \int_0^t \int_0^{\widehat{r}^{(k)}_s} W^{\tau_k}(ds,du)
+ \sigma_Z \int_0^t \int_0^{\widehat{r}^{(k)}_{s-}}  \int_{\mathbb{R}^+} \zeta \widetilde{N}^{\tau_k}
(ds,du,d\zeta).
 \eeqnn
Here $(W^{\tau_k},N^{\tau_k})$ is the translator of $(W,N)$ at $\tau_k$, i.e. 
$W^{\tau_k}([0,t]\times A)=W([\tau_k,\tau_k+t]\times A)$ and $N^{\tau_k}([0,t]\times A\times C)=
N([\tau_k,\tau_k+t]\times A\times C)$.  Consider (\ref{poisson immigration}), we have easily that $Y_t=\widehat{r}^{(0)}_t$ for $0\leq t< \tau_1$ 
and similarly $Y_t=\widehat{r}^{(1)}_{t-\tau_1}$ for $\tau_1\leq t<\tau_2$. More generally, applying Proposition \ref{prop:decomposition}, 
it is not hard to see that $\overline{r}^{(1)}_t:=\widehat{r}^{(1)}_t-\widehat{r}^{(0)}_{\tau_1+t}$ is independent of 
$\{\widehat{r}^{(0)}_t\}$ and have the same distribution as $\{r_t^c\}$. 
Thus for $\tau_1\leq t<\tau_2$, $Y_t^x=r_t^x+\overline{r}^{(1)}_{t-\tau_1}$.
Similarly, for $\tau_{k-1}\leq t<\tau_k$, $Y_t^x=r_t^x+\sum_{i=1}^{k} \overline{r}^{(i)}_{t-\tau_i}$, where $\overline{r}^{(k)}_t
=\widehat{r}^{(k)}_t-\widehat{r}^{(k-1)}_{\tau_k-\tau_{k-1}+t}$. Thus we have 
(\ref{excursion}). Also by Step 1 and the exponential formula,
\begin{eqnarray*}
\mathbb{E}\left[e^{-qY^x_t}\right]&=&\exp\left\{-xv(t,p)+\lambda_J\int_0^t\left(1-e^{-cv(t-s,p)}\right)ds\right\} \\
&=&\exp\left\{-xv(t,p)+\lambda_J\int_0^t\left(1-e^{-cv(u,p)}\right)du\right\}.
\end{eqnarray*}
The last equality follows from replacing $t-s$ by $u$.

{\bf Step 3: limit using the renormalization of the auxiliary process.} 
Consider a sequence of $Y^{(n)}$ defined by (\ref{poisson immigration}), where $J_t$ replaced by $J^{(n)}_t$  with parameter $\lambda_n=abn$ and
$c$ replaced by $c_n=1/n$. Let $r$ given by \eqref{lambda-integral} with initial value $x$. It is not hard to see that 
 $Y^{(n)}\rightarrow r$ in law as $n\rightarrow\infty$.
Then 
\beqnn
\mathbb{E}\left[e^{-pr_t}\right]=\lim_{n\rightarrow\infty}\mathbb{E}\left[e^{-pY^{(n)}_t}\right]=\exp\left\{-xv(t,p)-ab\int_0^t v(s,p)ds\right\}
\eeqnn
\qed

\end{document}